\documentclass[12pt]{article}
\usepackage[utf8]{inputenc}
\usepackage{fullpage}
\usepackage{authblk}

\usepackage{amssymb}
\usepackage{fullpage}
\usepackage[round]{natbib}
\usepackage{cancel}

\usepackage[algo2e,ruled]{algorithm2e}
\usepackage{algorithm}
\usepackage[noend]{algpseudocode}

\usepackage{amsmath}
\usepackage{amsthm}
\usepackage{latexsym}
\usepackage{graphicx}
\usepackage{color}
\usepackage{graphics}
\usepackage{xspace}

\usepackage{mathtools}
\usepackage[mathscr]{euscript}
\usepackage{hyperref}
\usepackage{tikz}
\usetikzlibrary{decorations.pathreplacing,decorations.markings}

\DeclareMathAlphabet\mathbb{U}{msb}{m}{n}
\usepackage{xpatch}

\DeclareMathOperator*{\E}{\mathbb E}

\DeclarePairedDelimiter{\abs}{\lvert}{\rvert}

\newtheorem{theorem}{Theorem}[section]

\newtheorem{lemma}[theorem]{Lemma}

\theoremstyle{definition}
\newtheorem{definition}[theorem]{Definition}

\theoremstyle{remark}

\newtheorem*{example*}{Example}

\newcommand{\cI}{\mathcal{I}}

\newcommand{\sA}{{\mathscr A}}

\newcommand{\sC}{{\mathscr C}}

\newcommand{\sE}{{\mathscr E}}
\newcommand{\sF}{{\mathscr F}}

\newcommand{\sH}{{\mathscr H}}

\newcommand{\sL}{{\mathscr L}}

\newcommand{\sU}{{\mathscr U}}

\newcommand{\high}{\textsf{high}\xspace}
\newcommand{\low}{\textsf{low}\xspace}

\newcommand{\mfloor}{\textsf{floor}\xspace}
\newcommand{\mceil}{\textsf{ceiling}\xspace}

\newcommand{\threshold}{\textsf{threshold}\xspace}

\newcommand{\NOR}{\textsf{NOR}\xspace}
\newcommand{\NOT}{\textsf{NOT}\xspace}

\newcommand{\R}{\mathbb R}
\newcommand{\C}{\mathbb C}

\newcommand{\ignore}[1]{}

\makeatletter
\newtheorem*{rep@theorem}{\rep@title}
\newcommand{\newreptheorem}[2]{%
\newenvironment{rep#1}[1]{%
 \def\rep@title{#2 \ref{##1}}%
 \begin{rep@theorem}}%
 {\end{rep@theorem}}}
\makeatother

\hypersetup{
  colorlinks   = true,
  urlcolor     = blue,
  linkcolor    = blue,
  citecolor    = blue
}

\newreptheorem{theorem}{Theorem}
\newreptheorem{lemma}{Lemma}
\newreptheorem{corollary}{Corollary}
\newreptheorem{proposition}{Proposition}

\usepackage{verbatim}
\usepackage[T1]{fontenc}
\usepackage[tt=false,type1=true]{libertine}
\usepackage[varqu]{zi4}
\usepackage[libertine]{newtxmath}

\title{Complex Dynamics in Autobidding Systems}
\author{Renato Paes Leme, Georgios Piliouras, Jon Schneider, Kelly Spendlove, Song Zuo}

\date{}

\affil{Google \\ \texttt{\{renatoppl,gpil,jschnei,spendlove,szuo\}@google.com}}

\newcommand{\cameraready}[2]{#2}

\begin{document}

\maketitle

\begin{abstract}
  It has become the default in markets such as ad auctions for participants to bid in an auction through automated bidding agents (autobidders) which adjust bids over time to satisfy return-over-spend constraints. Despite the prominence of such systems for the internet economy, their resulting dynamical behavior is still not well understood. Although one might hope that such relatively simple systems would typically converge to the equilibria of their underlying auctions, we provide a plethora of results that show the emergence of complex behavior, such as bi-stability, periodic orbits and quasi periodicity. We empirically observe how the market structure (expressed as motifs) qualitatively affects the behavior of the dynamics. We complement it with theoretical results showing that autobidding systems can simulate both linear dynamical systems as well logical boolean gates.
\end{abstract}

\section{Introduction}

In recent years it has become increasingly common for buyers
to participate in markets such as internet advertising through automated bidding agents (autobidders). Instead of bidding directly,
they provide high level goals (such as budgets and return-over-spend goals) to the autobidders
who then optimize auction bids on behalf of advertisers to  satisfy their goals. In the past years,
the research community has tried to bridge the gap between the widespread use of such bidding agents in practice and the
little theoretical understanding we had on the behavior of such agents.

Since \citet{aggarwal2019autobidding}, this has resulted in a vibrant research direction studying the design of auctions for autobidders \cite{aggarwal2023multi,balseiro2021landscape,balseiro2022optimal,balseiro2023optimal,golrezaei2021auction,lu2023auction,lv2023auction,lv2023utility,xing2023truthful}. Common research themes include work on improving common auction formats for autobidders \citep{balseiro2021robust,deng2021towards,deng2023autobidding,deng2023individual,mehta2022auction}, analyzing the equilibria of autobidders \cite{li2023vulnerabilities,liu2023auto} and proving price of anarchy, social welfare guarantees for autobidders \cite{deng2022efficiency,deng2023efficiency,liaw2023efficiency,liaw2023efficiency2,lucier2023autobidders,fikioris2023liquid} in standard auction formats (see also Section~\ref{sec:related} for other recent work on autobidding systems).

The focus of the community has been so far in understanding the equilibria of such systems as well
as the average behavior of its dynamics. This is a very natural and useful viewpoint to take and indeed
it is the norm in Economics to focus on the notion of equilibrim. Here we will taking the complementary, dynamical systems
viewpoint and try to describe the qualitative behavior of autobidder dynamics.  
Thus we will be focusing on 
questions of whether and how they reach equilibrium? Do they display periodic, quasi-periodic and chaotic
behavior? Can complex behavior emerge even in relatively simple systems?

\paragraph{Dynamic Viewpoint} Although traditional game theory focuses on Nash equilibria and similar behavioral fixed points of these dynamics, the actual day-to-day behavior of the agents can be significantly more complicated and thus should be an object of careful consideration.
A notable proponent of the dynamical viewpoint in Economics is Stephen Smale who won the Fields Medal for his contributions in topology. \citet{smale1976dynamics} argues that ``equilibrium theory is far from satisfactory'' since it often ignores how equilibrium is reached and assumes that agents perform long-range optimization. He also comments on how some parts of mathematical Economics require deep tools from topology such as fixed point theorems, which often obscure important underlying phenomena. In the same paper, Smale proposes trying to approach the same problems from the perspective of calculus and differential equations.

The dynamic viewpoint has been particularly successful in establishing the emergence of complex behavior of game playing dynamics and characterizing their limit non-equilibrium behavior~\citep{BaileyEC18,galla2013complex,vlatakis2020no,andrade2021learning,andrade2023no,hsieh2020limits,mertikopoulos2017cycles,cheung2021online,sanders2018prevalence,piliouras2023multi,chotibut2020route,bailey2020finite,palaiopanos2017multiplicative,bielawski2021follow}. Such results are typically established for playing general-sum normal-form games which allow for very rich classes of utility functions. Auctions, on the other hand, are more structured and well-behaved games which are carefully designed to minimize strategic behavior. In fact, a number of recent results 
\cite{balseiro2019learning,feng2021convergence,kolumbus2022auctions,bichler2023computing} establish the convergence of bidding dynamics based on learning algorithms in standard auction formats such as first and second price auctions\footnote{The picture becomes more complex in the presence of budget constraints \cite{chen2023complexity} and non-uniform tie-breaking rules \cite{chenpeng2023complexity}, where the equilibrium computation becomes PPAD hard, suggesting more complex dynamic behavior.}.
Although these results suggest a relatively straightforward view of the dynamical nature of bidding behavior, they refer to traditional quasi-linear agents and not the autobidding model discussed in the first two paragraphs. When we reconsider this question in the autobidding model, a completely different picture emerges. Before describing our results, we discuss the autobidding model.

\paragraph{Autobidding systems} An automated bidding agent has a couple of different tasks:  (i) predict the value of each item (e.g. predict the probability that an ad click will lead to a product purchase) (ii) forecast inventory and price competition; (iii) optimize bids to hit a certain return-over-spend (ROS) or budget target. Let's consider a concrete scenario in internet advertisement in which an autobidder competes in a pay-per-click second price auction. An advertiser cares about \emph{`conversion events'} which for the purposes of our example we can consider as the user purchasing the advertiser's product. The customer then will tell the autobidder to maximize conversions (purchases) subject to paying at most \$2/conversion.

For task (i) above, the autobidder $i$ will first build an ML model that predicts for each click $j$ what is the probability $v_{ij}$ of a conversion given a click using past data provided by the advertiser. For task (ii) the autobidder will construct functions $x_{ij}(b)$ and $p_{ij}(b)$ which predict the expected number of clicks and expected payments given bids $b$. Having access to $v_{ij}, x_{ij}, p_{ij}$ the autobidder is ready to solve task (iii) which consists in maximizing the expected number of conversions $\sum_{j} v_{ij} x_{ij}$ subject to the total expected payment $\sum_{j} p_{ij}$ being at most \$2 per conversion. In other words, for the target $\tau = 0.5 \cdot $conversions/dollar the problem solved by the autobidder is:

$$\max_{b_{ij}} \sum_{j} v_{ij} x_{ij}(b)\quad  \text{ s.t. } \quad \tau_i \sum_j p_{ij}(b) \leq \sum_{j} v_{ij} x_{ij}(b)$$
In this paper, we will assume that the autobidder has already a perfect model of values,  inventory and competition $v_{ij}, x_{ij}, p_{ij}$ and focus on task (iii). Furthermore, we will assume that we have a repeated game where $v_{ij}, x_{ij}, p_{ij}$ are the exactly the same across time. The only thing changing from period to period are the bids.

At this point, it is worth noting that the autobidder's objetive function represents a departure from the usual quasi-linear model in Economics. While there is no consensus on why this sort of objective is so widespread in practice, we offer some possible explanations.

\begin{itemize}
\item Credibility: advertisers directly observe total spend and total value (measured in terms of conversions) and hence can directly verify that the autobidder is hitting their target ROS. In contrast, it is hard for advertisers to verify the autobidder behavior in the quasi-linear model is indeed optimal without knowing the full $x_{ij}, p_{ij}$ auction.
\item Cross-channel bidding: Advertisers tend to bid on multiple channels (different internet platforms, TV, placing ads on printed media, ...) and try to optimize across all channels. A natural metric is to move more budget to larger ROS channels until the ROS is equalized. The astute reader may object here that to optimize ROS cross-channel, the advertisers don't need to equalize ROS across channels, but rather equalize the cost of the "more expensive" conversion \cite{deng2023multi,aggarwal2023multi}. While this is certainly true in theory, such metric is both: (a) very noisy since it depends on a single item instead of an average over many items; (b) not available in many channels like TV or printed media.
\item Alignment with business goals: business tend to evaluate the effectiveness of their marketing departments via high level aggregated metrics such as ROS, so it is natural for advertisers to align their business strategies 
\end{itemize}
We refer to \citep{alimohammadi2023incentive,feng2023strategic,perlroth2023auctions} for a more in-depth discussion about the economic rationale behind autobidders' objective function.

\paragraph{Autobidding Dynamics} We studying a dynamical system capturing the competition of multiple bidders bidding for multiple items in typically a second price auction through autobidders. Each buyer's input to their autobidder is a return over spend (ROS) target $\tau_i$ and this is fixed over time.
The goal of each autobidder is to
choose their bids in order to maximize allocation value subject to the given ROS constraint. As is common in these applications we will assume that the autobidders apply uniform bid scaling, i.e., the bid profile for each bidder is a multiplicative fraction of their respective valuation profile for all items. Thus, each  each autobidder controls a single parameter: the bid multiplier and increases/decreases them at a rate that is equal to the slack/oversaturation of the ROS target. This is a natural and tractable model that enjoys a number of desirable properties such as connections to PID controllers, as well as individual optimality guarantees such as guaranteed satisfaction of ROS constraints in a time-average sense, a.o.    The full details as well as more detailed discussion about the model can be found in Section $\ref{sec:model}$.

\paragraph{Methodology} Methodologically, the paper is a combination of modeling, empirical analysis and theory. In Section \ref{sec:model} we model autobidding as a dynamical system represented by a differential equation. We proceed with an empirical investigation in Section \ref{sec:qualitative} where we construct and simulate instances with increasingly complex behavior (drawing an analogy to synthetic biology). In Sections \ref{sec:linear_system} and \ref{sec:circuits} we switch to a theoretical approach where we formally prove properties of such systems. In particular, we show that other types of complex systems (linear dynamics and boolean circuits) can be simulated by autobidding systems.

\paragraph{Our Results and Techniques}  First, we show that in the special case of two bidders, an autobidding  system always converges to equilibrium (Theorem \ref{thm:convergence}). The proof is based on the celebrated Poincare-Bendixson theorem that constrains the possible dynamics in two dimensional systems, and additionally uses the Poincare-Bendixson criteria, which is upheld by the ROS system. 
Interestingly, even in this constrained setting we establish the possibility of multiple attracting equilibria as well as unstable equilibria. Next, we present a series of constructions where the dynamics of the corresponding markets have increasingly complex behavior. In Section \ref{sec:3_oscillation}, we present a three bidder system that leads to oscillations by converging to a limit cycle, showing that our previous convergence result is in some sense as strong as possible. 
In Sections \ref{sec:motifs} to  \ref{sec:quasi}, we showcase the complexity of autobidding dynamics by establishing formal connections between them and repressilator dynamics in synthetic
biology, which are generated by a genetic regulatory network consisting of at least one negative feedback loop. This allows for inducing predictable but complex behavioral patterns such as bistability, limit cycles and quasi-periodicity based on embedding bidding repressilator-inspired networks with specific combinatorial properties (e.g. directed cycles of odd/even length, hierarchical networks, etc). We also study the impact of the auction formats by comparing the behavior of autobidding dyanmics and second price, first price and mixtures thereof.

Next, we go beyond studying specific classes of examples and argue properties about the general complexity and descriptive power of autobidding dynamics. In Section \ref{sec:linear_system}, we show that we can simulate solutions of arbitrary linear dynamical systems via  autobidding dynamics. The result is a consequence of a result in linear algebra on the Jordan decomposition of purely-competitive matrices that can be of independent interest. Finally, in Section \ref{sec:circuits}, we show that autobidding systems can exhibit a different type of complex behavior: they can simulate arbitrary boolean circuits. The result follows from a complexity-style result where we construct gadgets composed of bidders and items to simulate the behavior of logical gates.


\subsection{Other Related Work}
\label{sec:related}

\paragraph{Autobidding Dynamics and Equilibrium Existence} 
There are several closely related works. \citet{lucier2023autobidders} analyze the PoA and the bidder regret along the bidding response dynamics, while not answering whether such dynamics converge or not. \citet{liu2023auto} proposes a bidding algorithm where the iterated best response process converges to an equilibrium under certain assumptions. 

The existence of equilibrium in the autobidding system is first studied by \citet{aggarwal2019autobidding}, who show the existence of pure Nash equilibrium in truthful auctions assuming smooth environment. \citet{li2023vulnerabilities} further prove that the pure Nash equilibrium always exists for second price auction even when the environment is non-smooth (e.g., with point mass distributions) by including tie-breaking strategies as part of the equilibrium definition. Nevertheless, it remains unknown whether the autobidding dynamics can converge to such equilibrium or whether such equilibrium is stable (say, contracting within its neighborhood). \citet{li2023vulnerabilities} also provide a PPAD-hardness of computing equilibrium which offers evidence that the dynamics may not converge or converge very slowly. 

\paragraph{Other Autobidding Topics}
Another closely relevant topic is the design of bidding algorithms for autobidders \cite{balseiro2023joint,deng2023multi,feng2023online,golrezaei2021bidding,he2021unified,liang2023online,lucier2023autobidders,susan2023multi}. The goal of these design, however, is usually either to theoretically prove good bounds for autobidder regrets or social welfare, or to empirically show good performance. Whether the ROS system converges under such bidding algorithms and the properties of the resulting dynamics are less concerned. On the auction side, \citep{deng2022posted} studies the design of pricing policies in the presence of autobidders.





\section{Setting: Autobidder ROS Systems}
\label{sec:model}

This section models the dynamics of how bids vary over time in an autobidding ROS system as differential equation \eqref{eq:diff_equation_multipliers} defined later in this section.

We will study a system where buyers interact with an auction through automated bidding agents (autobidders). The buyer's input to their autobidder is a return over spend (ROS) target $\tau_i$ and those will be fixed over time. The autobidders then submit bids into the central auction and those bids will change dynamically over time. We will assume that in any given moment $t$ in time, the auction will be allocating $k$ items, so each autobidder $i$ will submit bids $b_{ij}(t)$ for $j=1..k$ for each of those items.  The time $t$ will be a continuous variable taking values in $[0, \infty)$.

\begin{figure}[h]
\centering
\begin{tikzpicture}[scale=1.25]
\draw (0,0) rectangle (3,2) node[pos=.5] {Central Auction};

\draw (-3.5,0) rectangle (-1.5,.7) node[pos=.5] {Autobidder $n$};
\node at (-2.5, 1) {$\cdots$};
\draw (-3.5,1.3) rectangle (-1.5,2) node[pos=.5] {Autobidder $1$};

\draw[->] (-1.5,.35)-- node[above]{$b_{nj}$} (0,.35);
\draw[->] (-1.5,1.65)--node[above]{$b_{1j}$} (0,1.65);

\draw (-7,0) rectangle (-5,.7) node[pos=.5] {Buyer $n$};
\node at (-6, 1) {$\cdots$};
\draw (-7,1.3) rectangle (-5,2) node[pos=.5] {Buyer $1$};

\draw[->] (-5,.35)-- node[above]{$\tau_{n}$} (-3.5,.35);
\draw[->] (-5,1.65)--node[above]{$\tau_{1}$} (-3.5,1.65);
\end{tikzpicture}
\label{fig:lower_bound}
 \end{figure}

 \paragraph{Auction} The auction will take bids $b_{ij}$ as inputs and produce allocations $x_{ij} \in [0,1]$ such that $\sum_i x_{ij} = 1$ and payments $p_{ij} \geq 0$ for each buyer $i$ and item $j$. Unless stated otherwise the auction will be a pure second price auction with uniform random tie-breaking, i.e.:
$$x_{ij}(b) = \frac{\mathbf{1}\{ b_{ij} = \max_{s} b_{sj} \}}{\abs{\{h;  b_{h} = \max_{s} b_{sj} \}}} \qquad \qquad p_{ij}(b) = x_{ij}(b) \cdot \max_{s\neq i} b_{sj} $$

\paragraph{ROS Objective} We will assume that the buyer's return over spend (ROS) $\tau_i$, the number of items and the agents' values $v_{ij}$ for each item $j$ are fixed over time. The goal of each autobidder is to choose their bids in order to maximize allocation value $\sum_j v_{ij} x_{ij}$ subject to the ROS constraint.

\[\max_{b} \sum_j v_{ij} x_{ij}(b_{ij}) \quad \text{s.t.} \quad \tau_i \cdot \sum_j p_{ij}(b_{ij}) \leq \sum_j v_{ij} x_{ij}(b_{ij})\]

Our first observation is that if we define $\tilde{v}_{ij} = v_{ij} / \tau_i$ then we can re-write the problem above as:

$$\max_{b_{ij}} \sum_j \tilde{v}_{ij} x_{ij}(b) \quad \text{s.t.} \quad \sum_j p_{ij}(b) \leq  \sum_j \tilde{v}_{ij} x_{ij}(b)$$

Hence by assuming that the values $v_{ij}$ already incorporate the ROS targets, we can assume w.l.o.g. that $\tau_i = 1$ for all advertisers. Hence we will focus on the problem:

\begin{equation}\label{eq:autobidder_goal}
\max_{b_{ij}} \sum_j v_{ij} x_{ij}(b) \quad \text{s.t.} \quad U_i(b) \geq 0 \quad \text{where } U_i(b):= \sum_j v_{ij} x_{ij}(b) - p_{ij}(b)\end{equation}
where $U_i(b)$ is the traditional quasi-linear utility associated with that allocation. It is important to highlight that while we track $U_i(b)$ the goal of autobidders is not to maximize $U_i(b)$ but rather to maximize value subject to utility being non-negative.

\paragraph{Uniform bid-scaling} Despite the fact that the auction is a second price auction, truthful bidding $b_{ij} = v_{ij}$ is no longer optimal under the objective function in equation \eqref{eq:autobidder_goal}. Consider the following example with $4$ items where we show the value of bidder $i$ and the maximum bid of other agents for each of the items:

\begin{center}
\begin{tabular}{ l |c c c c }
 $v_{ij}$ & $3$ & $5$ & $7$ & $6$ \\
 \hline
 $\max_{h \neq i} b_{hj}$ & $2$ & $5$ & $8$ & $12$ \\
\end{tabular}
\end{center}

A truthful bid $b_{ij} = v_{ij}$ wins the first $2$ items giving total value $8$ and $U_i(b) = 1 > 0$. Recall that the goal of autobidders is \emph{not} to maximize the  $U_i$ but rather to maximize value subject to the ROS constraint that $U_i(b) \geq 0$. A better solution is to bid $b_{ij} = (7/6) v_{ij}$, winning the first $3$ items, obtaining total value $15$ with $U_i(b) = 0$.

A natural strategy used in practice for this problem is called uniform bid-scaling, which consists in bidding:
$$b_{ij} = m_i \cdot v_{ij}$$
for a bid multiplier $m_i \in [1, \infty)$. The prevalence of uniform bid-scaling in applications has to do with the fact that it is the optimal bidding strategy when the contribution of each item goes to zero. We call this setting the \emph{smooth limit} which we will describe in detail later in this section. In the smooth limit, the bidding problem becomes equivalent to the fractional knapsack problem  \citep{feldman2007budget}. Feldman et al also show that uniform bid-scaling is approximately optimal even outside the smooth limit -- both with a theoretically established approximation guaranteed and a much better guaranteed observed in practical instances.

\paragraph{Autobidding Dynamics}

From this point on, we will assume that each autobidder controls a single parameter: the bid multiplier $m_i$. We will abuse notation and write $x_{ij}(m), p_{ij}(m), U_i(m)$ as a function of vector of multipliers $m = (m_1, \hdots, m_n)$. Clearly the allocation value $\sum_j x_{ij} v_{ij}$ is monotone in the multiplier, so the goal of each autobidder is to find the highest multiplier that satisfies the constraint $U_i(m) \geq 0$. The function $U_i(m)$ is increasing for $m \leq 1$ since any new item we win when we change the multiplier from $m$ to $m+\delta$ has price less than the value. For $m\geq 1$, any new item we win when we change the multiplier from $m$ to $m+\delta$ has price larger than the value and hence $U_i$ is decreasing. (See Figure \ref{fig:U_shape}).

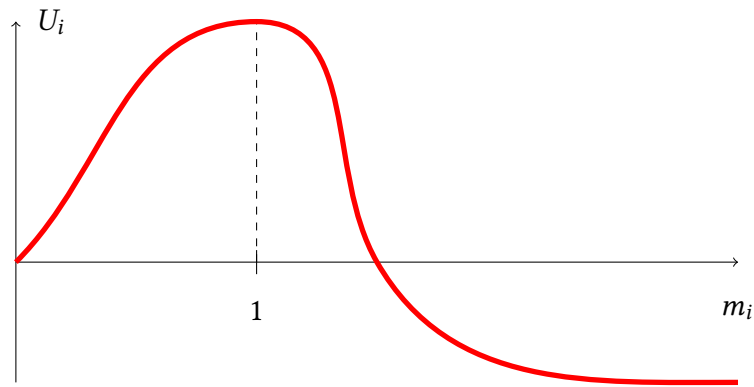
\begin{figure}[h]
\centering
\begin{tikzpicture}[scale=1.6]
\draw[->] (0,0) -- (6,0);
\draw [dashed] (2,0) -- (2,2);
\draw[->] (0,-1) -- (0,2);
\draw (2,-.1) -- (2,.1);
\node at (2,-.4) {$1$};
\node at (6,-.4) {$m_i$};
\node at (.3,2) {$U_i$};
\draw[line width=2, red] (0,0) to[out=45,in=180] (2,2) to[out=0,in=120] (3,0)  to[out=-60,in=180] (6,-1);
\end{tikzpicture}
\caption{The function $U_i$ is non-decreasing for $m_i \leq 1$ and non-increasing for $m_i \geq 1$.}
\label{fig:U_shape}
 \end{figure}

From that discussion it should be clear that the bid multiplier that optimizes the goal in equation \eqref{eq:autobidder_goal} is the value $m_i \geq 1$ such that $U_i(m_i)$ is as close as possible to zero. This suggests the following simple bidding strategy:
\begin{itemize}
\item whenever $U_i(m) > 0$ increase the multiplier since there is slack in the constraint so there is room to increase the value.
\item whenever $U_i(m) < 0$ decrease the multiplier since we are paying more we should per the ROS constraint.
\end{itemize}
A practical method for adjusting multiplier is what is typically called in electrical engineering as as PID controller. In this method, we adjust the control variable (here the bid multiplier) proportionally to the slack of the violation of the constraint. This is successfully applied in practice from the temperature control in HVAC systems all the way to bidding and budget pacing in ad systems \cite{balseiro2023field,tashman2020dynamic,yang2019bid,zhang2016feedback,smirnov2016online}. Recently, PID controllers were shown to be an instance of mirror descent applied to a suitable loss function \cite{balseiro2023analysis}\footnote{This result builds a bridge from PID controllers to more traditional learning algorithms. It remains an interesting open question to understand the performance of autobidding under a broader class of update rules derived from mirror descent or FTRP. We note, however, that autobidding is not a traditional online learning problem since it has a hard ROS constraint thus standard no-regret algorithms can’t be applied out of the box.}.

When we translate this control strategy to a differential equation in the space of multipliers, what we get is the following set of equations:

\begin{equation}\label{eq:diff_equation_multipliers}
\frac{d m_i}{dt} = U_i(m), \forall i
\end{equation}

There are variation\footnote{} of this equation that satisfy the same goal such as $\frac{d}{dt} (\log m_i) = U_i(m)$ or $\frac{d}{dt}  m_i = U_i(m)^\alpha$ that are also often used in practice and achieve a similar effect. Here we will study the linear version in \eqref{eq:diff_equation_multipliers}. The linear version has the property that if the vector of multipliers $m(t)$ evolves over time according to the equation above and the multipliers stay bounded within $m_i \in [0,B]$ then:

$$\frac{1}{T}\int_0^T U_i(m(t)) dt = \frac{1}{T}\int_0^T \frac{d m_i}{dt} dt = \frac{m_i(T) - m_i(0)}{T} \leq \frac{B}{T}$$
In other word, the control loop satisfies the ROS constraints on average up to a $O(1/T)$ error.

\paragraph{Smooth limit} It will often simplify our analysis to assume that the each item has an infinitesimal contribution to the total utility $U_i(m)$. In this limit the number of items $k \rightarrow \infty$ and the value of each item $v_{ij} \rightarrow 0$. Another (perhaps more intuitive) way to describe this limit is to keep the number of items $k$ fixed but instead of assuming a fixed value $v_{ij}$ we assume that the bidders' valuation are described by a random vector $v = (v_{ij})_{i=1..n, j = 1..k}$ with positive and $C^1$-density everywhere on a box $[0,M]^{nk}$. Under such a regime the utilities $U_i(m)$ are defined as $\E[\sum_j v_{ij} x_{ij}(m) - p_{ij}(m)]$ and are $C^1$-functions of the vectors of multipliers.


\section{Basics of differential equations and dynamical systems}

In this section we review a few basic ideas about differential equations and dynamical systems.  We refer the reader to standard references \cite{arnold1992ordinary, hirsch2012differential, perko2013differential, strogatz2018nonlinear} for more details.  Given a function $U:\R^n\to \R^n$, consider the system of differential equations $\frac{dx}{ dt} = U(x)$. Such a system is called \emph{autonomous} since the right hand side has no explicit dependence on time. A \emph{solution} to this system is a function $x:J\to \R^n$ for interval $J\subset \R$, such that for all $t\in J$, $\frac{d}{dt}x(t) = U(x(t))$.  Geometrically, $x(t)$ is a curve in $\R^n$ whose tangent vector exists for each $t\in J$ and equals $U(x(t))$ at $t$. An initial condition for a solution $x:J\to \R^n$ is a $t_0\in J$ and $x_0\in \R^n$ such that $x(t_0)=x_0$.  For simplicity, it is normally assumed $t_0 = 0$.  The curve which $x$ traces out in $\R^n$ is called the \emph{orbit} through $x_0$.  If $U$ is $C^1$, then the existence and uniqueness theorem guarantees that given an initial condition $x(0) = x_0$ there exists such a solution, and moreover it is only such solution.

In general, obtaining explicit solutions to differential equations is an extremely difficult task. Instead, one focuses on a qualitative understanding of the behavior of the set of solutions.  The \emph{flow} generated by a differential equation is a continuous map $\phi:\R \times \R^n\to \R^n$ which represents the set of solutions, i.e., $\phi(t,x_0)$ is the value at time $t$ of the solution with initial value $x_0$.  A flow satisfies the properties i) $\phi(0,x) = x$ for all $x\in \R^n$ and $\phi(s,\phi(t,x))=\phi(s+t,x)$ for all $s,t\in \R$ and $x\in \R^n$.

We say that a point $x \in \R^n$ is an equilibrium whenever $U_i(x) = 0$ for all $i$ and hence the constant function $x(t) = x$ is a solution the equation. If there exists a $\tau>0$ such that $\phi(t+\tau,x_0)=\phi(t,x_0)$ for all $t$ and $x_0$ is not an equilibrium, then the solution $\phi_t(x_0)$ is called a \emph{periodic orbit}. The smallest such $\tau$ is called the \emph{period}.  A set $S$ is \emph{invariant} (with respect to $\phi$) if $\phi(t,S)\subset S$ for all $t\in \R$; equilibria and periodic orbits are examples of invariant sets.  A set $S$ is \emph{positively invariant} (with respect to $\phi$) if $\phi(t,S)\subset S$ for all $t\geq 0$.

A point $x\in \R^n$ is an $\omega$-limit point for the solution through $x_0$ if there is a sequence $t_n\to \infty$ such that $\lim_{n\to \infty}\phi(t_n, x_0)\to x$. The $\omega$-limit set of $x_0$, denoted $\omega(x_0)$, is the set of all $\omega$-limit points of $x_0$.  If $\omega(x_0)$ consists of a single point $x$, then $x$ is necessarily an equilibrium.  If, on the other hand, $\omega(x_0)$ is a periodic orbit with period $\tau$ then the solution $x(t)$ through $x_0$ will eventually oscillate with a period approaching $\tau$.  Thus informally, if $x\in \omega(x_0)$ then the solution through $x_0$ eventually reaches $x$ (possibly infinitely often), as it accumulates at $x$ as time advances. Similarly to $\omega$-limit sets, we define the $\alpha$-limit points and the $\alpha$-limit set $\alpha(x_0)$, replacing $t_n\to \infty$ with $t_n\to -\infty$.  By \emph{limit set} we mean either an $\omega$ or $\alpha$ limit set.

The Poincare-Bendixson theorem  characterizes limit sets for planar differential equations, i.e., $\frac{d}{dt}(x,y) = U(x,y)$, with $U:\R^2\to \R^2$, and roughly says that a nonempty $\omega$ limit set is either an equilibrium, a periodic orbit, or a finite number of equilibria together with orbits connecting them.  For further discussion see, e.g., \cite[Section 3.7]{perko2013differential}.  In the special case that the Bendixson criteria holds, i.e., $\frac{\partial U_1}{\partial x} + \frac{\partial U_2}{\partial y}\neq 0$, it can be shown that in fact an omega limit set is either empty or a single equilibrium, see \cite{mccluskey1998stability}.


\begin{theorem}[Poincare-Bendixson Refinement]\label{thm:poincare_bendixson}
If $\frac{\partial U_1}{\partial x} + \frac{\partial U_2}{\partial y}\neq 0$ on $\R^2$ and $\omega(x_0,y_0)$ is nonempty, then $\omega(x_0,y_0)$ is an equilibrium.
\end{theorem}

\section{Qualitative behavior of ROS systems}\label{sec:qualitative}

Our goal is to understand the behavior of an ROS system defined by the differential equation \eqref{eq:diff_equation_multipliers}. One of the most basic questions to ask is: what is the long-term behavior of the system trajectories?  It is natural to hope that in these systems all trajectories converge to equilibria.  However, if the system does not equilibriate, then what is its asymptotic behavior?


We'll start with some general facts about ROS systems.
First, we give two characterizations of the form of the utility function $U_i$. All missing proofs can be found in Appendix \ref{appendix:proofs_qualitative}.

\begin{lemma}\label{lemma:ros:UI}
In the smooth limit, we have that $U_i$ is $C^1$ and: (1) $\frac{\partial U_i}{\partial m_i} > 0 \text{ for } m_i < 1$; (2)  $\frac{\partial U_i}{\partial m_i} < 0 \text{ for } m_i > 1$; (3) $U_i(m) \geq 0$ for $m_i = 1$.
\end{lemma}




Suppose $\phi$ is the flow for an ROS system with $n$ bidders given by equation \eqref{eq:diff_equation_multipliers}.  The next lemma shows that once flow preserves the region $[1,\infty)^n$, i.e., if $m(0)\in [1,\infty)^n$ is an initial condition, then $m(t) \in [1,\infty)^n$.

\begin{lemma}\label{lemma:invariance_above_1}
 The set $[1,\infty)^n\subset \R^n$ is positively invariant with respect to $\phi$.
\end{lemma}

For the remainder of the paper, we restrict our attention to the behavior of ROS systems on $[1,\infty)^n$.

\subsection{Two Bidder System}\label{sec:two_bidders}

We consider an instance with two bidders and two items where the values $v_{ij}$ are given by the following matrix (recall that targets are normalized to $1$). In all examples here, we assume $i$ indexes rows and $j$ indexes columns:
$$\begin{bmatrix}
2 & 1 \\ 1 & 2
\end{bmatrix}$$
The pairs of multipliers that are in equilibrium are $(2,2)$ and $(1,x), (x,1)$ for $x \geq 2$. Those are highlighted in blue in the right side of Figure \ref{fig:2_bidder}. In the left side of the same figure, those corresponds to points where $(U_1, U_2) = (0,0)$. Along the diagonal $m_1 = m_2$, the dynamics is attracted by the symmetric equilibrium $(2,2)$ as expected. However, even a slight deviation from the symmetric equilibrium causes the dynamics to converge towards one of the two extremal equilibria $(2,1)$ or $(1,2)$.

\tikzset{
  on each segment/.style={
    decorate,
    decoration={
      show path construction,
      moveto code={},
      lineto code={
        \path [#1]
        (\tikzinputsegmentfirst) -- (\tikzinputsegmentlast);
      },
      curveto code={
        \path [#1] (\tikzinputsegmentfirst)
        .. controls
        (\tikzinputsegmentsupporta) and (\tikzinputsegmentsupportb)
        ..
        (\tikzinputsegmentlast);
      },
      closepath code={
        \path [#1]
        (\tikzinputsegmentfirst) -- (\tikzinputsegmentlast);
      },
    },
  },
  mid arrow/.style={postaction={decorate,decoration={
        markings,
        mark=at position .5 with {\arrow[#1]{stealth}}
      }}},
}

\begin{figure}[h]
\begin{center}
\hfill
\includegraphics[scale=.36]{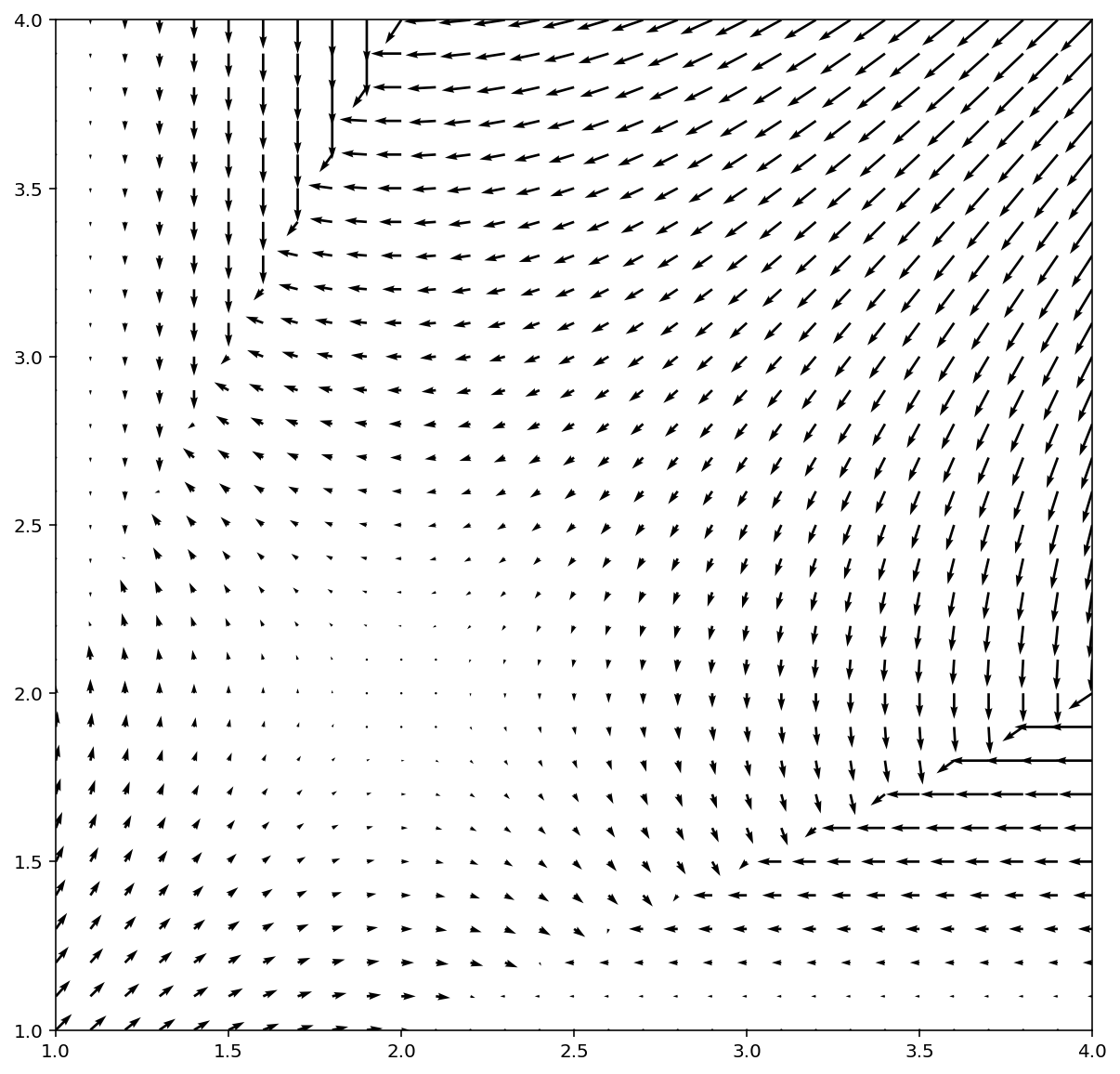}
\hfill
\begin{tikzpicture}[scale=2.25]
\draw (1,1) rectangle (4,4);
\node at (1,.9) {$1$};
\node at (2,.9) {$2$};
\node at (.9,1) {$1$};
\node at (.9,2) {$2$};
\draw[blue, line width=2pt] (1,2) -- (1,4);
\draw[blue, line width=2pt] (2,1) -- (4,1);
\draw[fill=blue] (2,2) circle [radius=1pt];
\draw[red,postaction={on each segment={mid arrow=red}}] (1,1) -- (2,2);
\draw[red,postaction={on each segment={mid arrow=red}}] (4,4) -- (2,2);
\input{2_bidder_orbits}
\end{tikzpicture}
\hfill
\end{center}
\caption{Vector field for two bidders [left]. The vector $(U_1(m), U_2(m))$ gives the direction of the dynamics at point $m$. Orbits of the dynamical system [right].}
\label{fig:2_bidder}
\end{figure}


In this example, regardless of the starting point, the dynamics always converges to an equilibrium. We will show that this is true in general for two bidders and any number of items.

We will assume that we are in the smooth regime, and so the $U_i$ are differentiable. Since the dynamic stays within $[1, \infty)^2$ (by Lemma \ref{lemma:invariance_above_1}) we can use that $\partial U_i / \partial m_i < 0$ for $m_i > 1$ in the smooth limit. 


\begin{theorem}
\label{thm:convergence}
Consider the case of the smooth limit ROS dynamics for two bidders. If $m_0\in (1,\infty]^2$ and the orbit through $m_0$ is bounded, then the orbit converges to an equilibrium.
\end{theorem}

\begin{proof}
Note that the Bendixson-Dulac criteria holds within $[1,\infty)^2 \subset \R^2$ from Lemma \ref{lemma:ros:UI}. Let $m_0\in [1,\infty)^2$. By Theorem \ref{thm:poincare_bendixson}, if $\omega(m_0)$ is nonempty, then it must be a single equilibrium.
\end{proof}

\subsection{Oscillations with three bidders}\label{sec:3_oscillation}

While systems of two bidders are guaranteed to converge, this is no longer true for systems of three or more bidders. We construct an example below that resembles a game of Rock-Paper-Scissors. Consider three bidders and three items with the following valuation matrix:

$$v = \begin{bmatrix}
2 & 1 & 0\\
0 & 2 & 1\\
1 & 0 & 2
\end{bmatrix}.$$

Notice that each item $j$ has two associated bidders: one with a high value (2) for $j$ and one with a low value (1). Dually, each bidder $i$ has two items it's interested in: an item of lower value, and an item of higher value. Consider an initial set of multipliers $m$ such that each item is won by its high value bidder.  In a second price auction, the price is then determined by the low value bidder.  Let us consider what happens as bidder $1$ increases their multiplier $m_1$:

\begin{itemize}
\item When $m_1$ increases, the price pressure on the second item increases, so $U_2$ decreases.
\item If $U_2$ decreases enough to become negative, then $m_2$ begins to decrease.  This in turn reduces the price pressure on the third item, so $U_3$ increases.
\item If $U_3$ increases enough to become positive, then $m_3$ begins to increase.  This in turn increases the price pressure on the first item, so $U_1$ decreases.
\item If $U_1$ decreases enough to become negative, it causes $m_1$ to decrease.  This in turn reduces price pressure on the second item, so $U_2$ increases.
\item If $U_2$ increases enough to become positive, $m_2$ increases. As $m_2$ increases, it increases the price pressure on the third item, so $U_3$ decreases.
\item If $U_3$ decreases enough to become negative then $m_3$ decreases. As $m_3$ decreases, it reduces price pressure on item 1 causing $U_1$ to increase, starting the cycle over again.
\end{itemize}

The overall effect is that increasing the bid multiplier $m_1$ generates a chain of effects through shared items and different bidders, ultimately leading to an oscillation.

This behavior can be obtained with fixed values for $v_{ij}$ but the dynamics exhibit various artifacts due to ties in the auction -- the dynamics causes bids on certain items to be the same and the path of multipliers depends on how those ties are resolved. In order to bypass this issue we pass to the smooth limit and replace the fixed valuations in the previous table by distributions so that ties occur with probability zero. For convenience we let the value distributions be given by beta distributions\footnote{The choice of a beta distribution is somewhat arbitrary. The goal is to have a smooth distribution over positive values that approximates a point mass. The choice of beta (instead of normal or log-normal) is to facilitate simulations. Since the density is a polynomial, it is possible to write the expected utilities $U_i$ in closed form instead of relying on sampling.} $\text{Beta}(a,b)$, parametrized by non-negative integers $a,b$ and with mean $a/(a+b)$ (see Figure \ref{fig:beta_distribution} in the appendix for example densities).

$$\begin{bmatrix}
\text{Beta}(2c,c) & \text{Beta}(c,2c) & 0\\
0 & \text{Beta}(2c,c) & \text{Beta}(c,2c)\\
\text{Beta}(c,2c) & 0 & \text{Beta}(2c,c)
\end{bmatrix}$$

Figure \ref{fig:3_bidder_beta_b} shows how the multipliers evolve over time given a certain starting point. Figure \ref{fig:3_bidder_beta_a} offers a different way to visualize the dynamics. We can think of the orbit $(m_1(t), m_2(t), m_3(t))$ as tracing a path in $\R^3$. We visualize it by projecting onto the first two dimensions and plotting the curve traced by $(m_1(t), m_2(t))$. In the right side of Figure \ref{fig:3_bidder_beta_a} we show three different orbits of the dynamic with the starting points chosen at random showing that they are all attracted by the same cycle, in other words, the periodic orbit is stable.

\begin{figure}[h]
\begin{center}
\includegraphics[width=\textwidth]{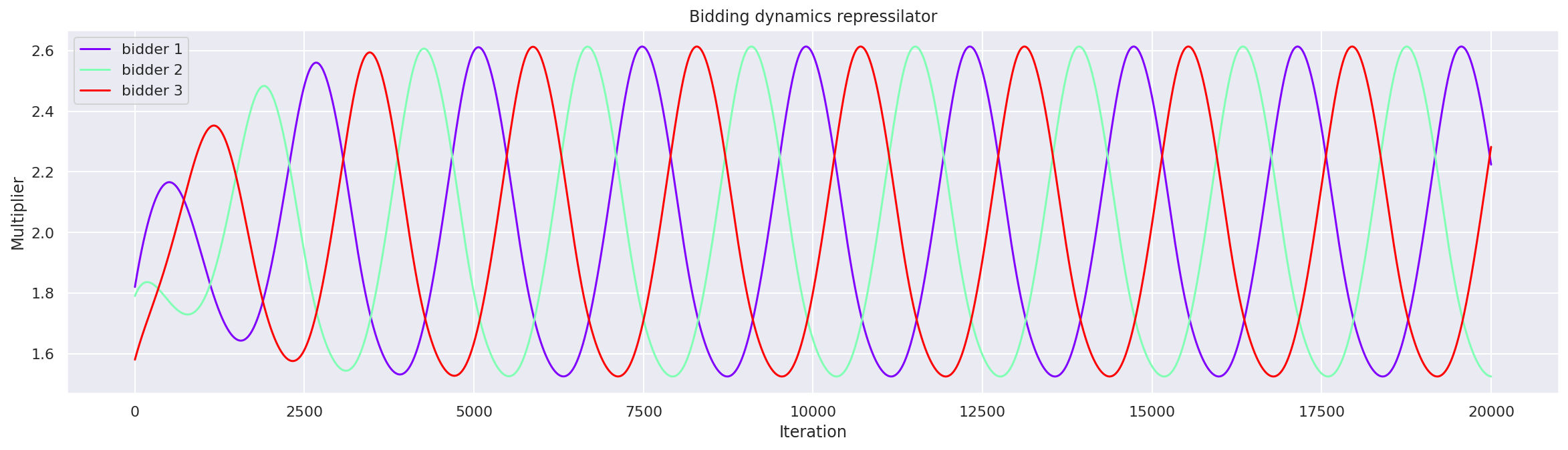}
\end{center}
\caption{The evolution of multipliers $m_1(t), m_2(t), m_3(t)$ over time.}
\label{fig:3_bidder_beta_b}
\end{figure}

\begin{figure}[h]
\begin{center}
\hfill
\includegraphics[width=0.48\textwidth]{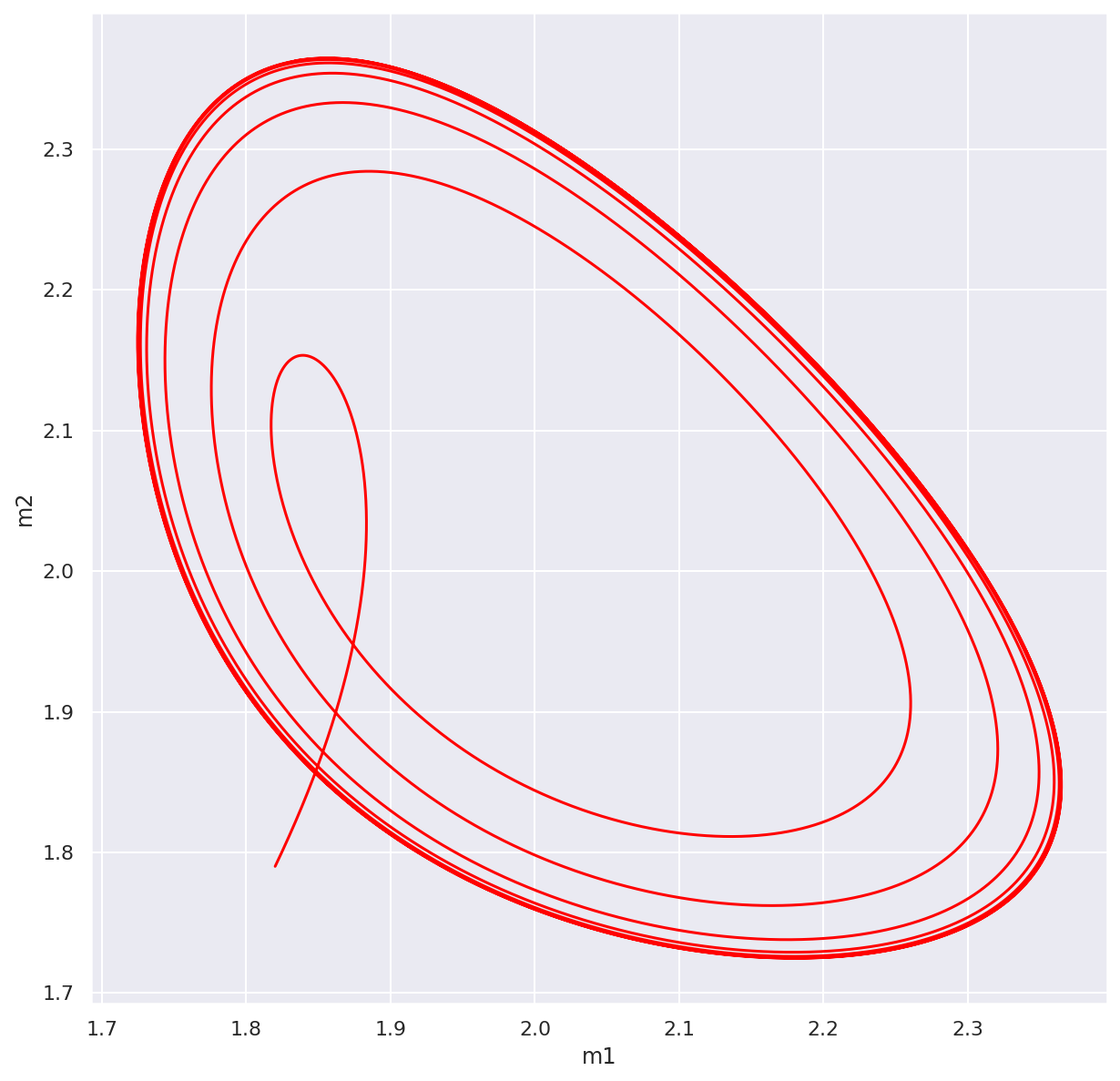}
\hfill
\includegraphics[width=0.48\textwidth]{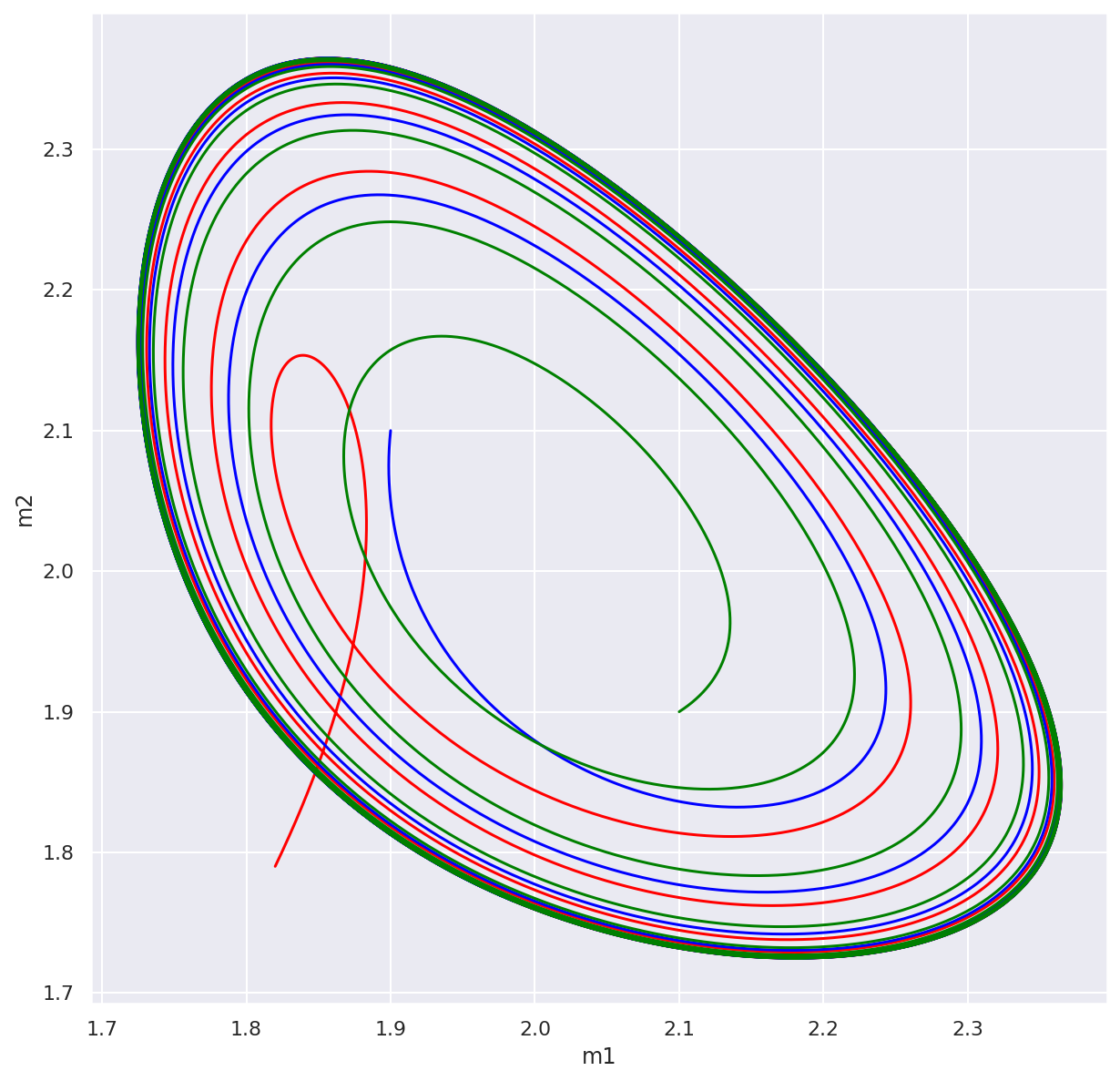}
\hfill
\end{center}
\caption{Left: The orbit of the multipliers projected on the first two multipliers $(m_1(t), m_2(t))$. Right: three orbits of the same dynamics with three different starting points. All the orbits are attracted by the same periodic cycle.}
\label{fig:3_bidder_beta_a}
\end{figure}

\subsection{Bidding Repressilators and Motifs}
\label{sec:motifs}

In the last subsection we constructed an example where the bids in an ROS system oscillate. This shows in particular that such systems don't necessarily converge. 
In this section we look to understand the nature of the oscillation through an analogy, and ask whether such systems can exhibit more complex behavior.
The analogy we draw is to that of the \emph{repressilator} in synthetic biology, a genetic regulatory network consisting of at least one negative feedback loop.  The original repressilator constructed in \cite{elowitz2000synthetic}  is a network of three genes arranged in a cycle of mutual inhibition, or \emph{repression}.  This arrangement can lead to oscillations under a similar principle that lead to the oscillations of \ref{sec:3_oscillation}.

Although the settings are very different (networks of interacting genes and proteins vs agents bidding in auctions) and the equations governing those phenomena are mathematically different, they share enough qualitative similarities that we believe we can borrow ideas from this line of work to understand phenomena in auction bidding. For that reason, we propose constructing \emph{bidding repressilators}.  In particular, we borrow two ideas from biology, the first being the notion of `repression', which will translate to a competitive interaction between bidders where one bidder `represses' another.  The second idea is that of a \emph{motif}: a particular recurrent graph or subgraph that  reflects the `design principles' underlying a network's behavior and function \cite{milo2002network, alon2019introduction}.  We will seek to understand the behavior of ROS systems by constructing simple motifs, which for example may exhibit oscillation, and then investigate how more complex behavior arises from extending or coupling the motif.

\paragraph{Bidding Repressilator} A bidding repressilator\footnote{While in synthetic biology repressilator is typically referred as a network of $3$ nodes, we here use the term more generally to refer to a network of any number of nodes where the directed edges mean that a node `represses' the other.} is a way to construct an ROS system from a graph of interactions between bidders. Given a directed graph, we will construct an ROS system that associates each node with a single bidder and each edge with a single item in such a way that an edge $a \rightarrow b$ means that bidder $a$ `\emph{represses}' bidder $b$ in the sense that (for the most part) an increase in the multiplier $m_a$ decreases $U_b$ and hence decreases the rate $\frac{d}{dt} m_b$. To achieve this effect we will create an item $j$ such that only bidders $a$ and $b$ have non-zero values for this item, given by:
$$v_{aj} \sim \text{Beta}(c, 2c) \qquad
v_{bj} \sim \text{Beta}(2c, c)$$
What will happen is that the value of $b$ for the item is in expectation twice the value of $a$ (Figure \ref{fig:beta_distribution}). Whenever their multipliers are not too far apart, typically $b$ will win the item and $a$ will be the price setter. Hence an increase of $m_a$ will lead to higher prices for $b$ which lowers $m_b$ achieving the desired effect. As an example, the graphs in Figure \ref{fig:repressilator_graphs}(a) and (b) show the two bidder and three bidder systems from Sections \ref{sec:two_bidders} and \ref{sec:3_oscillation}, respectively.

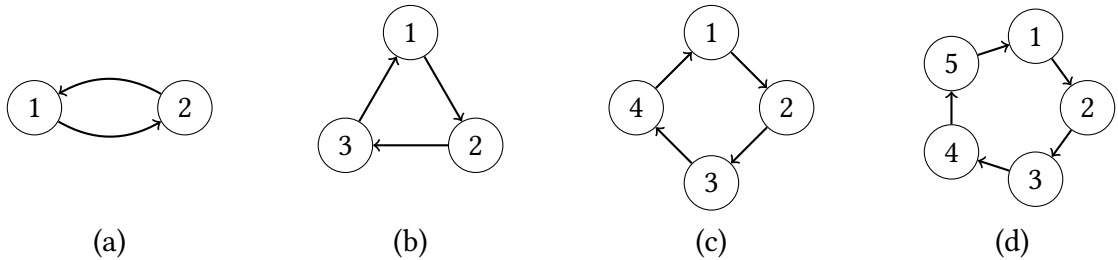
\begin{figure}[h]
\center
\begin{tikzpicture}[scale=1]

\begin{scope}[shift = {(-4, 0)}]
    \node[draw, circle] (A) at (180:1) {$1$};
    \node[draw, circle] (B) at (0:1) {$2$};

    \draw[->, thick] (B) to [bend right] (A);
    \draw[->, thick] (A) to [bend right] (B);
    \node at (0,-1.8) {(a)};
\end{scope}

\node[draw, circle] (A) at (90:1) {$1$};
\node[draw, circle] (B) at (330:1) {$2$};
\node[draw, circle] (C) at (210:1) {$3$};

\draw[<-, thick] (B) -- (A);
\draw[<-, thick] (C) -- (B);
\draw[<-, thick] (A) -- (C);
\node at (0,-1.8) {(b)};

\begin{scope}[shift = {(4, 0)}]
    \node[draw, circle] (A) at (90:1) {$1$};
    \node[draw, circle] (B) at (180:1) {$4$};
    \node[draw, circle] (C) at (270:1) {$3$};
    \node[draw, circle] (D) at (0:1) {$2$};

    \draw[->, thick] (B) -- (A);
    \draw[->, thick] (C) -- (B);
    \draw[->, thick] (D) -- (C);
    \draw[->, thick] (A) -- (D);
    \node at (0,-1.8) {(c)};
\end{scope}

\begin{scope}[shift = {(8, 0)}]
    \node[draw, circle] (A) at (72:1) {$1$};
    \node[draw, circle] (B) at (2*72:1) {$5$};
    \node[draw, circle] (C) at (3*72:1) {$4$};
    \node[draw, circle] (D) at (4*72:1) {$3$};
    \node[draw, circle] (E) at (0:1) {$2$};

    \draw[->, thick] (B) -- (A);
    \draw[->, thick] (C) -- (B);
    \draw[->, thick] (D) -- (C);
    \draw[->, thick] (E) -- (D);
    \draw[->, thick] (A) -- (E);
    \node at (0,-1.8) {(d)};
\end{scope}

\end{tikzpicture}
\caption{Examples of bidding repressilators with two through five bidders.}
\label{fig:repressilator_graphs}
\end{figure}

\paragraph{Caveats: Competitive Systems} The original repressilator in \citep{elowitz2000synthetic} is an example of a \emph{competitive system}, i.e., it has the form $\frac{d}{dt}x_i = f_i(x_1, \hdots, x_n)$ where $\partial f_i / \partial x_j \leq 0$ for all $j \neq i$. However, the bidding networks we construct (and ROS systems more generally) are not necessarily competitive, meaning that they may fail to satisfy $\partial U_i / \partial m_j \leq 0$ for $i \neq j$. The reason is that if we vary a bid multiplier $m_j$ in an ROS system there are two competing effects:
\begin{itemize}
\item increasing $m_j$ we increase the price of an item that $i$ wins, therefore reducing $U_i$.
\item increasing $m_j$ may cause bidder $i$ to lose an item it was winning. If a tiny increase in $m_j$ causes $i$ to lose item $k$, then before it was paying essentially their bid $m_i v_{ik}$ for that item and hence was obtaining negative utility $v_{ik} (1-m_i)$ for it assuming $m_i \geq 1$. Losing an item for which the utility was negative increases $U_i$.
\end{itemize}

Recall that the distributions of the values for an item is set in such a way that if there is an `repressor' edge $i \rightarrow j$, the values of $j$ are larger than the values of $i$ in expectation.  For agent $j$,  the first effect (the price-seller increases price pressure) dominates the second effect, unless the multiplier $m_i$ is very high. See Figure \ref{fig:U1_3cycle} [left] in the appendix. On the other hand, for agent $i$, the first effect is very small and dominated by the second effect, see Figure \ref{fig:U1_3cycle} [right].  It can be seen from Fig. \ref{fig:U1_3cycle} that the magnitude of the `reverse repression' direction is much smaller than the magnitude of the `repression'.





\subsection{Cyclic Feedback Bidding Repressilators}



The notion of a bidding repressilator encodes as a graph the structure of repression, or competition, between bidders. Our goal in this section is to examine to what extent the graph determines the behavior of the dynamics.  To do this we generalize the motif of a three node repressilator to an $n$ node cycle, which we call a \emph{cyclic feedback bidding repressilator}.  For examples, see Figure \ref{fig:repressilator_graphs}.

An interesting phenomena we observe is that:
\begin{enumerate}
\item Cycles of even length exhibit bistability (two stable equilibria), with half the bidders having a high multipliers and the other half having a low multiplier, as in Fig. \ref{fig:4_5_cycle} (top), while
\item cycles of odd length give rise to a stable periodic orbit, see Fig. \ref{fig:4_5_cycle} (bottom).
\end{enumerate}
One can readily see that the argument in Section \ref{sec:3_oscillation} extends to odd but not even number of bidders.  These observations match precisely the prediction for the so-called generalized repressilator (extending the three node cycle to a cycle of $n$ nodes), where stable periodic orbits only occur for an odd number of variables; whereas bistability occurs for an even number, with the two (stable) fixed points having even numbered nodes high, odd numbered nodes low, or vice versa \cite{muller2006generalized}.\footnote{One abstraction of the generalized repressilator is a \emph{monotone cyclic feedback system} i.e., a system arranged in a ring such that $\frac{d}{dt}x_i =  f_i(x_{i-1},x_i, x_{i+1})$ (and $x_0 = x_n$) and $\partial f_i / \partial x_{i+1}, \partial f_i / \partial x_{i-1} \leq 0$. In these systems a Poincare-Bendixson style theorem holds: the only invariant dynamics are periodic orbits and equilibria (and connections between), see \citet{mallet1990poincare, mallet1996poincare}.}   As we argue in the 'caveats' of the previous section, our system doesn't satisfy those properties exactly, but it does approximately. Empirically, we can observe that this is enough for the conclusions of that line of work to hold.



\begin{figure}[h]
\begin{center}
\includegraphics[width=\textwidth]{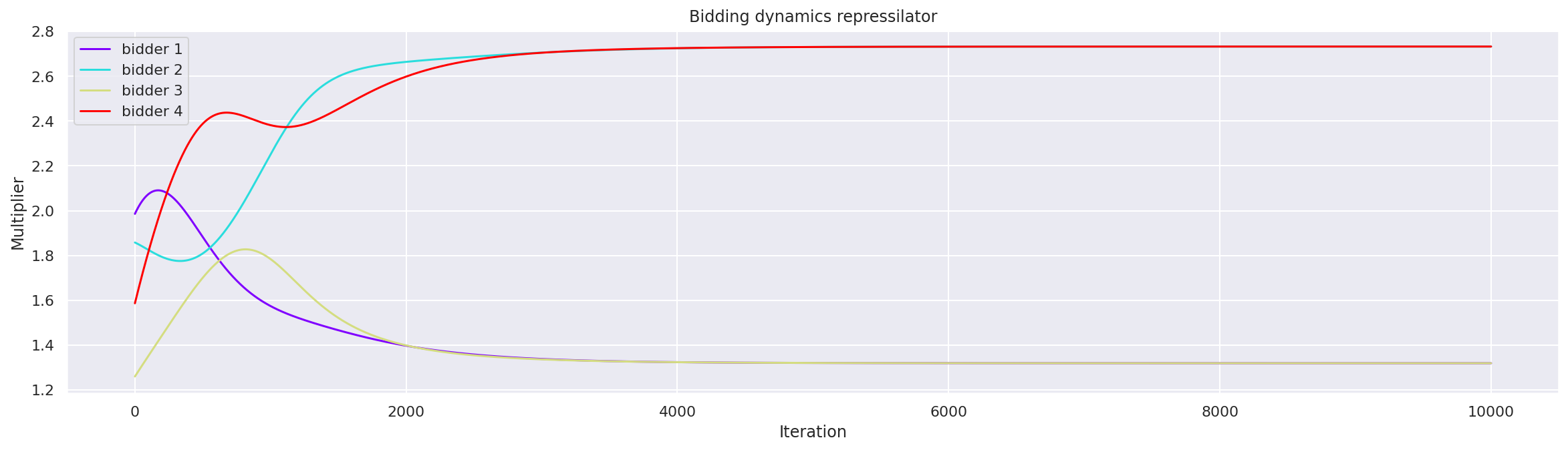}
\includegraphics[width=\textwidth]{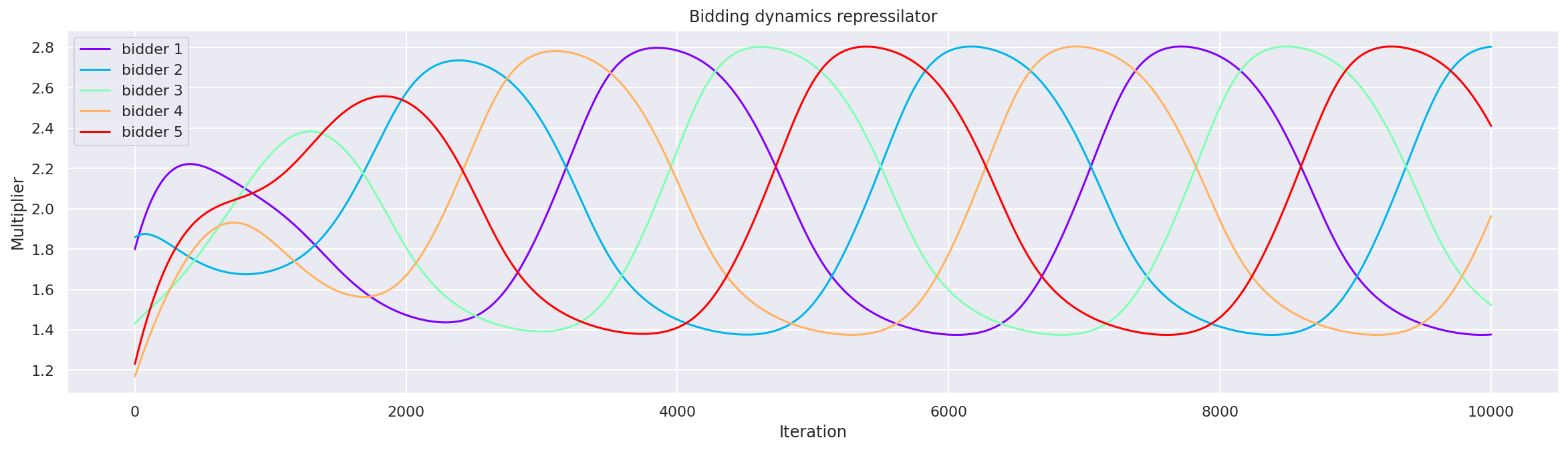}
\end{center}
\caption{A cyclic feedback system with four bidders is bistable, and for this initial condition converges to an equilibrium where the odd numbered bidders have low multipliers and the even numbered bidders have high multipliers. A cyclic feedback
 system with five bidders which has a stable periodic orbit.}
\label{fig:4_5_cycle}
\end{figure}



\subsection{Coupling of Repressilators and Quasi-periodicity}
\label{sec:quasi}

Using the analogy with the generalized repressilator (and in general monotone cyclic feedback systems) we argued that we should expect the behavior to mostly follow a periodic orbit if the bidders are arranged along a ring topology of a bidding repressilator. To obtain more complex behavior, we use the repressilator as a motif to introduce more complex topologies. In particular, a `repressilator of repressilators' where we arrange $9$ bidders in a way that each group of $3$ bidders acts as a super-node (as in Figure \ref{fig:repressilator_graphs}(b)). We simulate two different couplings described in Figure \ref{fig:repressilator_couplings}. In each of them, the first repressilator $123$ represses the second group $456$. And the groups $456$ and $789$ repress each other. What we observe is the emergence of quasi-periodic behavior. The orbits are no longer periodic, but rather trace a torus-like manifold in space.  Such behavior is called \emph{quasi-periodicty}, and is characterized by the system having two incommensurate periods, i.e., the ratio of the periods is irrational.

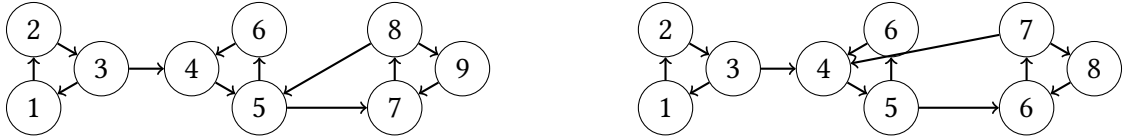
\begin{figure}[h]
\center
\begin{tikzpicture}[scale=0.6]

    \node[draw, circle] (3) at (0:1) {$3$};
    \node[draw, circle] (2) at (120:1) {$2$};
    \node[draw, circle] (1) at (240:1) {$1$};

    \draw[->, thick] (2) -- (3);
    \draw[->, thick] (1) -- (2);
    \draw[->, thick] (3) -- (1);

\begin{scope}[shift = {(4, 0)}]
    \node[draw, circle] (4) at (180:1) {$4$};
    \node[draw, circle] (5) at (300:1) {$5$};
    \node[draw, circle] (6) at (60:1) {$6$};

    \draw[->, thick] (4) -- (5);
    \draw[->, thick] (5) -- (6);
    \draw[->, thick] (6) -- (4);
\end{scope}

\begin{scope}[shift = {(8, 0)}]
    \node[draw, circle] (9) at (0:1) {$9$};
    \node[draw, circle] (8) at (120:1) {$8$};
    \node[draw, circle] (7) at (240:1) {$7$};

    \draw[->, thick] (9) -- (7);
    \draw[->, thick] (7) -- (8);
    \draw[->, thick] (8) -- (9);
\end{scope}

\draw[->, thick] (3) -- (4);
\draw[->, thick] (5) -- (7);
\draw[->, thick] (8) -- (5);

\begin{scope}[shift = {(14, 0)}]
    \node[draw, circle] (3) at (0:1) {$3$};
    \node[draw, circle] (2) at (120:1) {$2$};
    \node[draw, circle] (1) at (240:1) {$1$};

    \draw[->, thick] (2) -- (3);
    \draw[->, thick] (1) -- (2);
    \draw[->, thick] (3) -- (1);

\begin{scope}[shift = {(4, 0)}]
    \node[draw, circle] (4) at (180:1) {$4$};
    \node[draw, circle] (5) at (300:1) {$5$};
    \node[draw, circle] (6) at (60:1) {$6$};

    \draw[->, thick] (4) -- (5);
    \draw[->, thick] (5) -- (6);
    \draw[->, thick] (6) -- (4);
\end{scope}

\begin{scope}[shift = {(8, 0)}]
    \node[draw, circle] (9) at (0:1) {$8$};
    \node[draw, circle] (8) at (120:1) {$7$};
    \node[draw, circle] (7) at (240:1) {$6$};

    \draw[->, thick] (9) -- (7);
    \draw[->, thick] (7) -- (8);
    \draw[->, thick] (8) -- (9);
\end{scope}

\draw[->, thick] (3) -- (4);
\draw[->, thick] (5) -- (7);
\draw[->, thick] (8) -- (4);

\end{scope}

\end{tikzpicture}
\caption{Two different couplings of three $3$-cycle repressilators}
\label{fig:repressilator_couplings}
\end{figure}

\begin{figure}[h]
\begin{center}
\hfill
\includegraphics[width=0.48\textwidth]{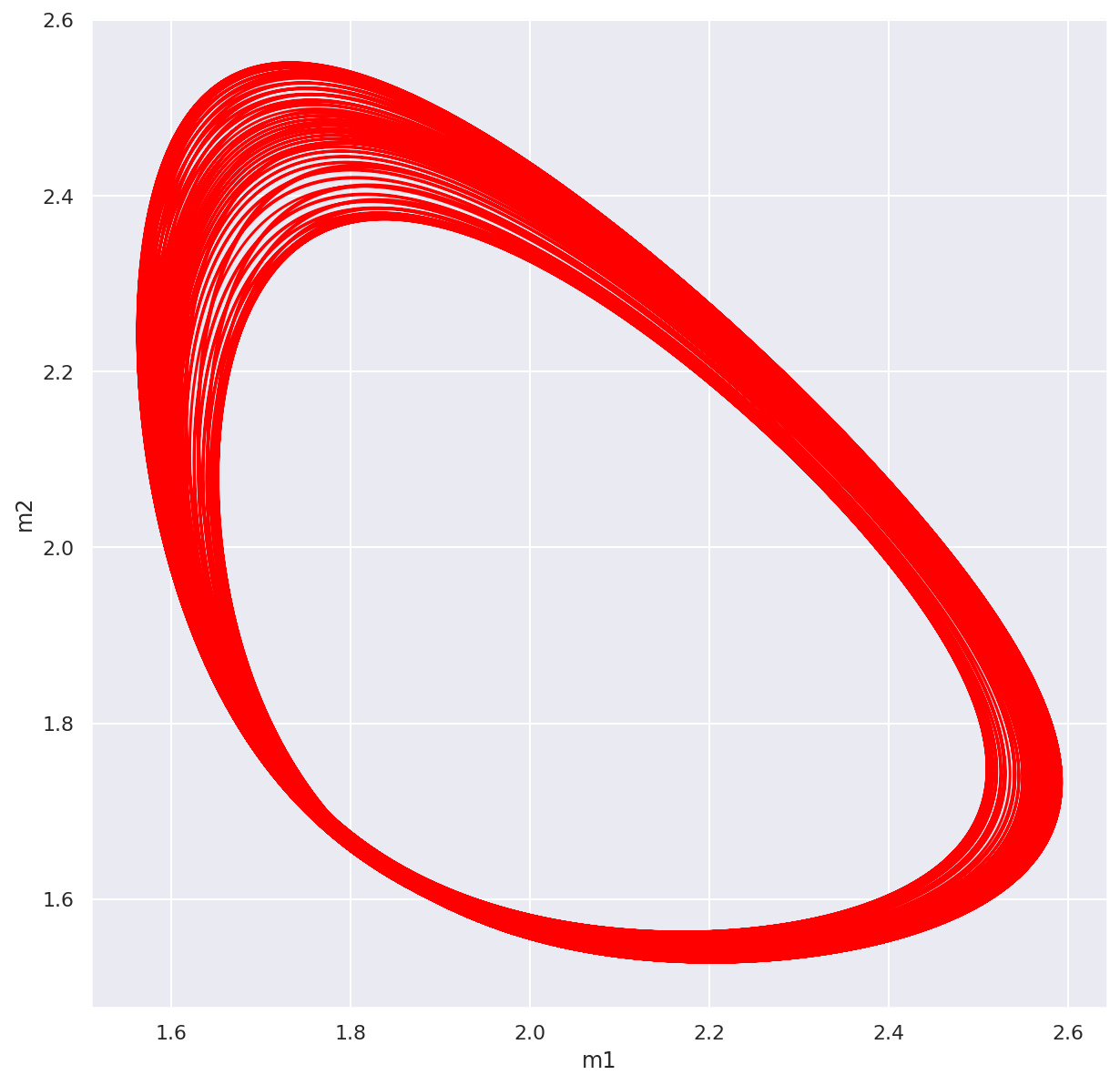}
\hfill
\includegraphics[width=0.48\textwidth]{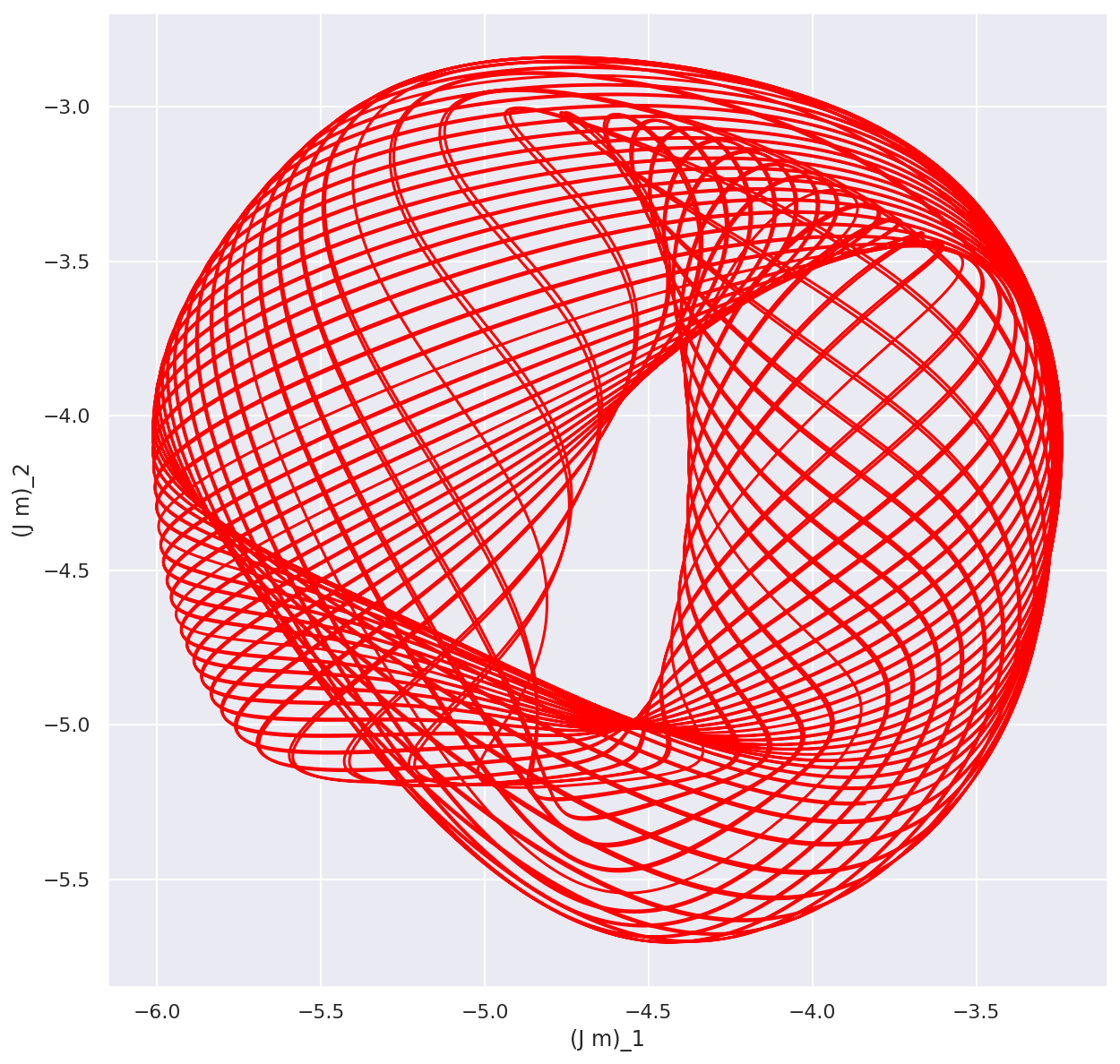}
\hfill
\end{center}
\caption{Quasi-periodic orbit corresponding to the second graph in Figure \ref{fig:repressilator_couplings}. The first graph corresponds to the orbit of multipliers $(m_1(t), m_2(t))$. The second graph is a random projection of the orbit $m(t)$.}
\label{fig:quasiperiodic_2}
\end{figure}

In figures \ref{fig:quasiperiodic_2} and \ref{fig:quasiperiodic_1} we show the orbit of the multipliers obtained from the dynamics of the repressilators described by the graphs in Figure \ref{fig:repressilator_couplings}. In each figure, the first plot corresponds to the path $(m_1(t), m_2(t))$ (i.e., the multipliers of the first two bidders) and the second graph is a path of $Jm(t)$ where $J$ is a random $2 \times 9$ matrix with i.i.d. standard Gaussian entries. Therefore, $J m(t)$ is a random projection of the $9$-dimensional orbit $m(t)$ to $2$ dimensions.

\subsection{Impact of the Auction Format}

In the last part of our empirical evaluation, we study the impact of the auction format on the dynamics. So far we have assumed that the underlying auction is a second price auction. Here we simulate the effect of changing the auction to a combination of a first and second price auction. We will consider a family of auctions parameterized by a parameter $\lambda \in [0,1]$ where the winner $i$ pays $\lambda \max_{s \neq i} b_{sj} + (1-\lambda) b_{ij}$. For $\lambda=1$ we have a pure second price auction and for $\lambda=0$ a pure first price auction. The rest of the model is kept unchanged. The bidders still use uniform bid-scaling and update multipliers according to $\frac{d m_i}{dt} = U_i(m)$.

For a particular dynamics (two coupled $3$-cycle repressilators), we plot in Figure \ref{fig:second_vs_first_price} for each value of $\lambda$ (in the x-axis) the range of values of the first multiplier $m_1$ after a long enough run of the dynamics (in the y-axis). For larger values of $\lambda$, the dynamics doesn't converge but we observe that the orbit gets more and more compressed the smaller $\lambda$ gets (see three last plots of Figure \ref{fig:orbits_first_second_price}). There is a phase transition around $\lambda=0.85$ where the dynamics converges. As the auction gets closer to a first price auction ($\lambda \rightarrow 0$), the equilibrium multipliers tend to $1$ as expected, reflecting the fact that in a first price auction with uniform bids scaling the dynamics converge regardless of the structure of the market. See the Appendix \ref{appendix:proofs_qualitative} for the proof of the following lemma, which is a dynamic version of a observation in \citet{deng2021towards} showing that a uniform bidding game in a first price auction with ROS maximization has an unique equilibrium that is efficient.

\begin{lemma}\label{lemma:first_price_dynamics} Under a first price auction in the smooth limit, the dynamics of an ROS-system converges the vector of multipliers $m_i = 1.0$ regardless of the market structure.
\end{lemma}

At first glance it seems remarkable that the behavior of autobidding dynamics under first price auctions is much simpler (convergence to multiplier $1.0$ regardless of the market topology) as compared to second price auctions that can exhibit all types of complex behavior. This is in part due to the choice of uniform bid-scaling as a bidding strategy, which is the optimal bidding strategy under second price (in the smooth limit) but is not guaranteed to be optimal under first price auctions. Hence the convergence result of autobidding dynamics in first price auctions is under a reasonable but potentially suboptimal bidding strategy.

\begin{figure}[h]
\begin{center}
\includegraphics[width=\textwidth]{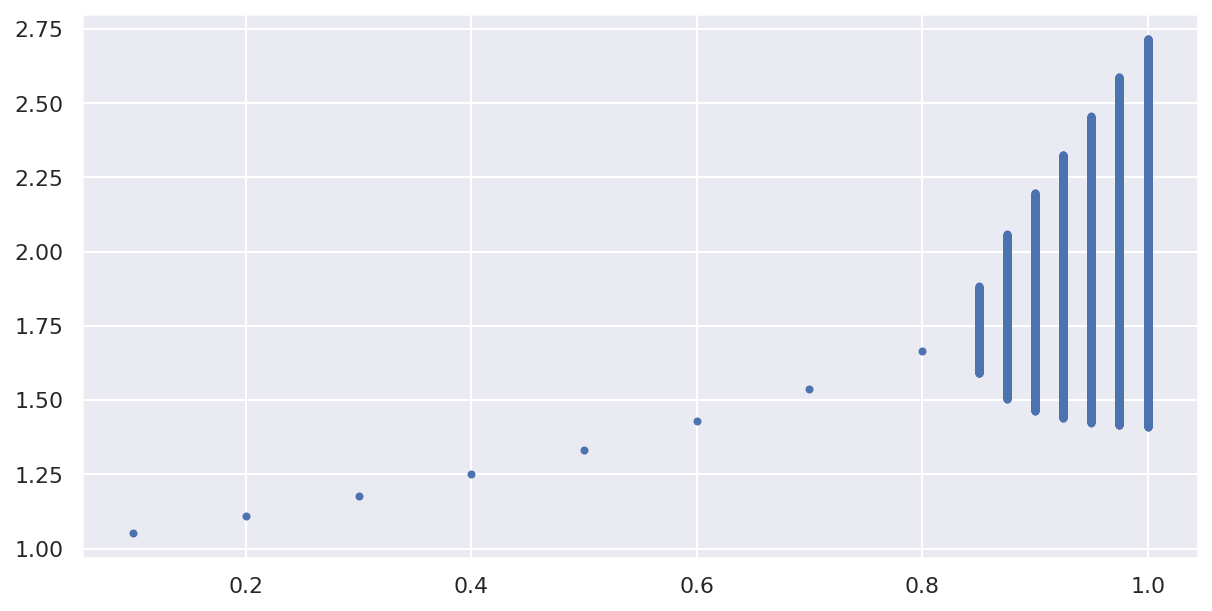}
\end{center}
\caption{Change in limit behavior as a function of the auction format. The x-axis corresponds to parameter $\lambda$ indicating that the auction run is $\lambda$ SPA + $(1- \lambda)$ FPA. The $y$-axis corresponds to the values of the multiplier $m_1$ observed in the limit. For $\lambda > 0.85$ the dynamics doesn't converge and the range of values correspond to the values of $m_1$ in the limit. For smaller values of $\lambda$ the dynamics converges which }
\label{fig:second_vs_first_price}
\end{figure}

\begin{figure}[h]
\begin{center}
\includegraphics[width=0.24\textwidth]{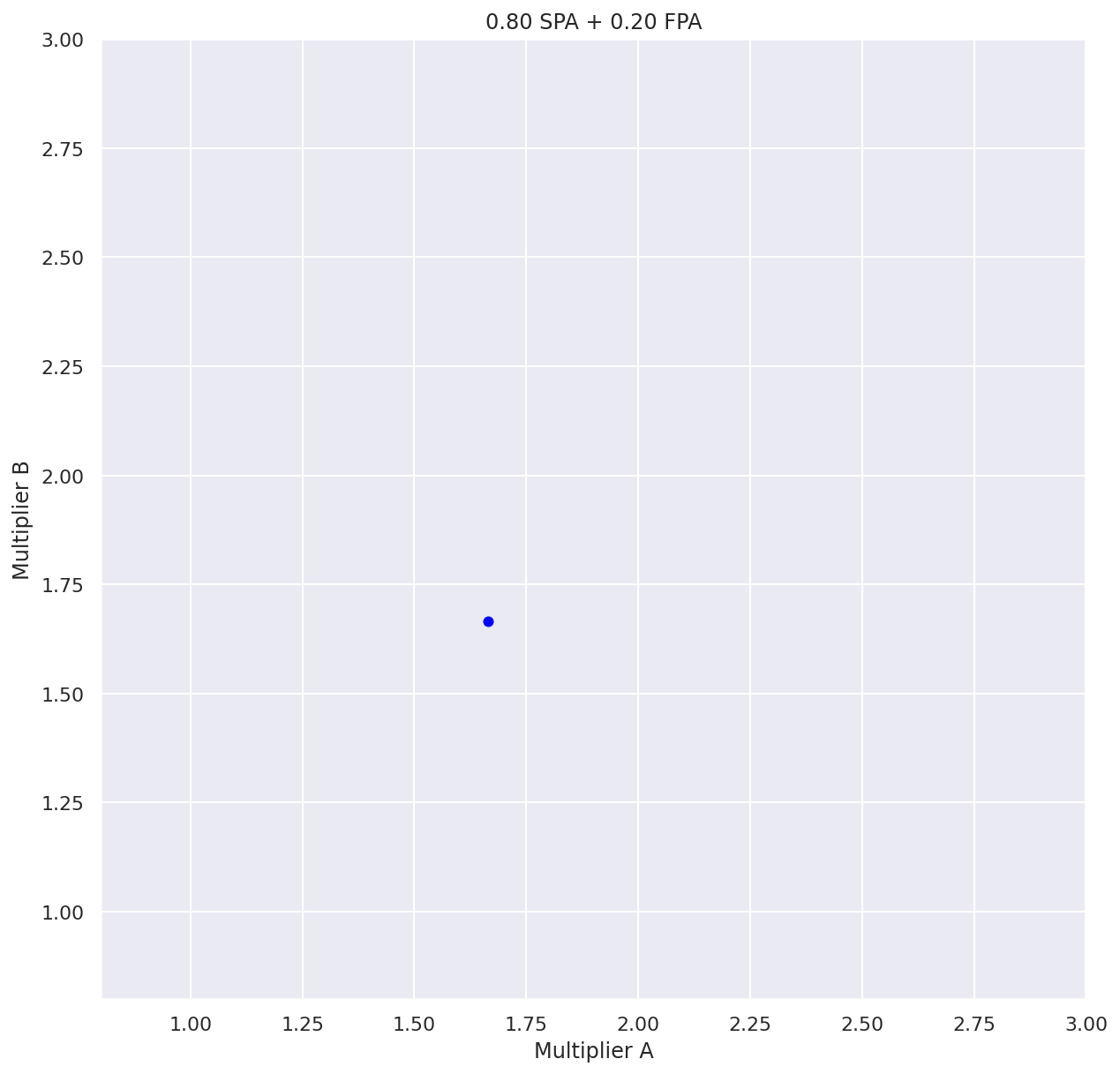}
\includegraphics[width=0.24\textwidth]{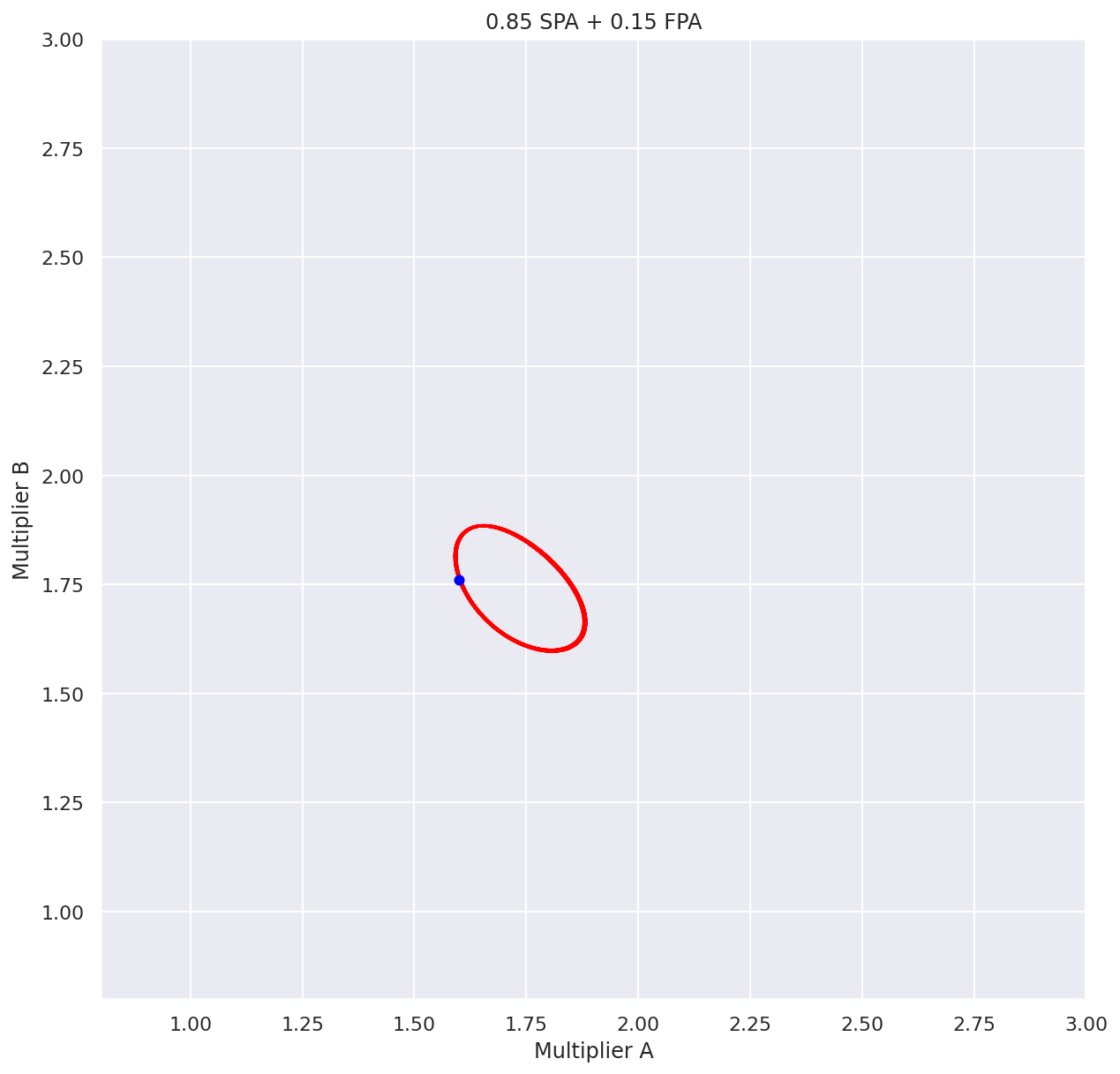}
\includegraphics[width=0.24\textwidth]{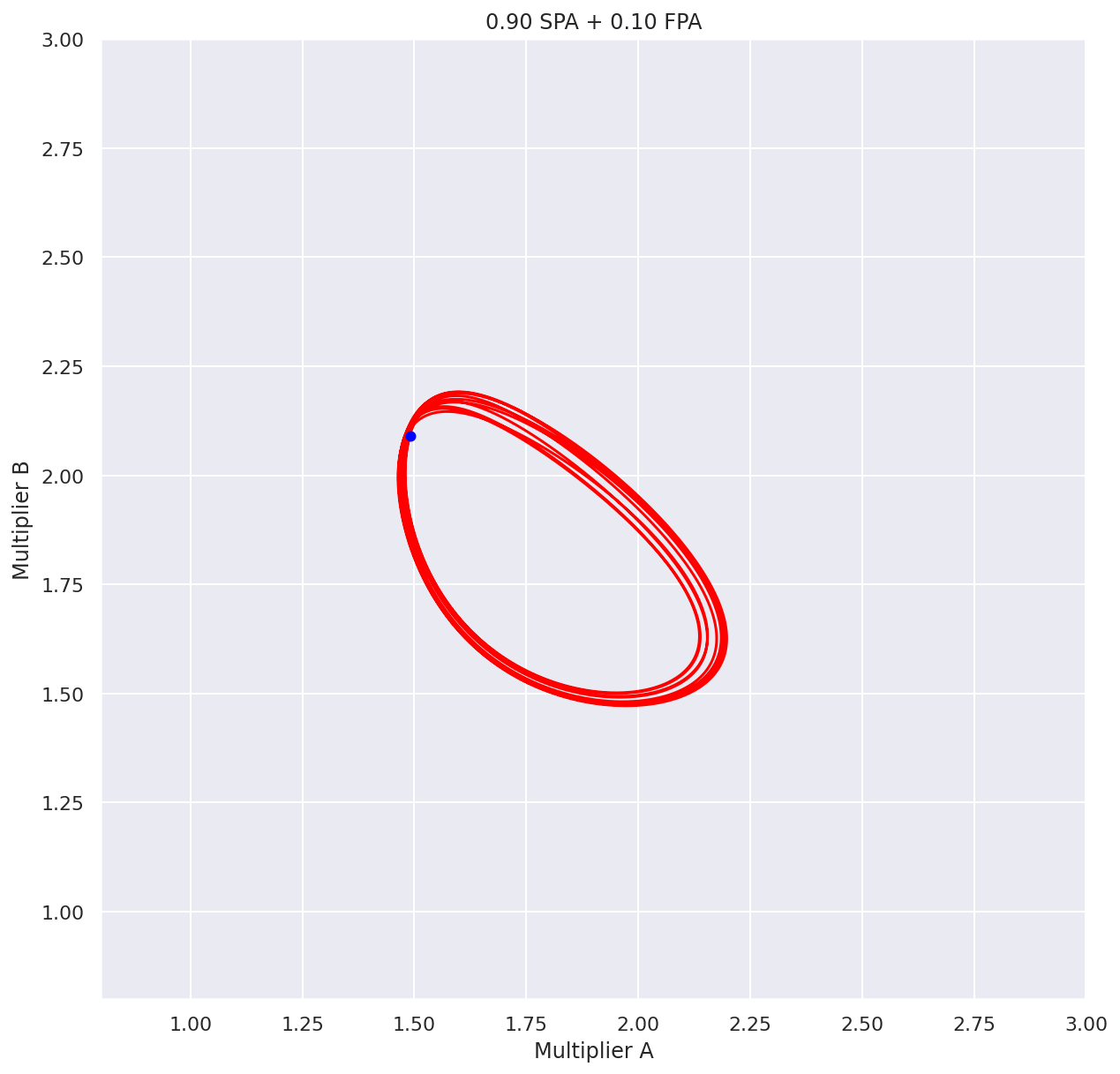}
\includegraphics[width=0.24\textwidth]{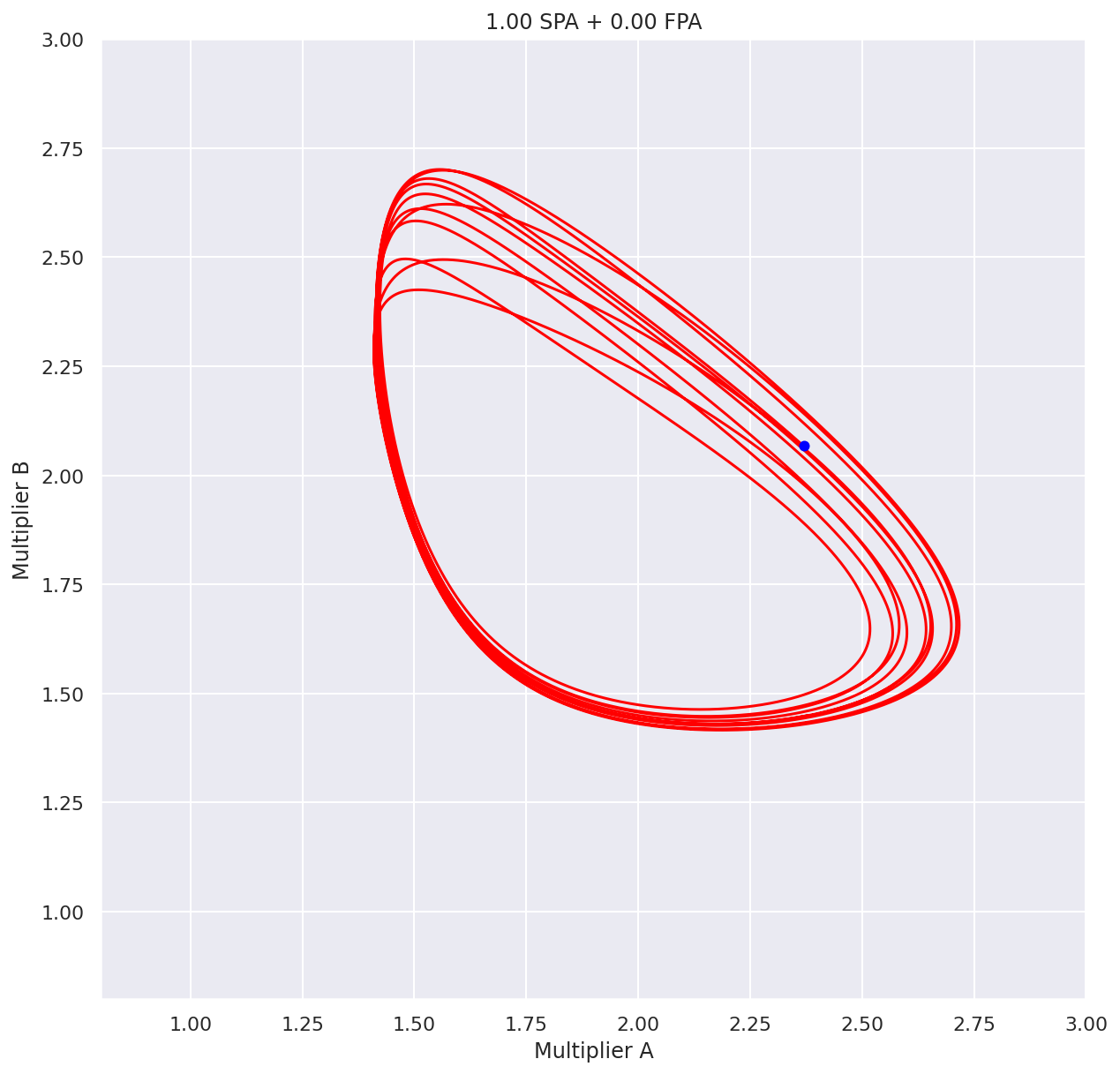}
\end{center}
\caption{Orbits of the ROS system when the auction is $\lambda$ SPA + $(1- \lambda)$ FPA for $\lambda \in \{0.8, 0.85, 0.9, 1.0\}$.}
\label{fig:orbits_first_second_price}
\end{figure}

\section{Simulating Linear Dynamical Systems}\label{sec:linear_system}

Thus far, we have empirically demonstrated the emergence of complex behavior in ROS systems by simulating those systems on networks constructed by coupling motifs. In particular, we have showed that ROS dynamics do not generally converge beyond $2$ bidders, and that they can exhibit complex dynamic behavior such as quasi-periodicity.

In this and the subsequent section we will change our methodology and formally prove some properties of the dynamics. Our first main result is to show that ROS systems can simulate the behavior of any linear dynamical system.

A linear dynamical system is a system of the type $\frac{d}{dt} x = Ax$ for $x \in \R^n$ and a $n \times n$ matrix $A$. Linear dynamical systems are very well understood dynamical systems and exhibit both periodic orbits and quasi-periodic behavior. By formally showing that ROS systems can simulate dynamical systems, we will formally show that ROS systems can exhibit such behaviors.

First, we define what it means for a system to simulate another. Given a differential equation $\frac{d}{dt} x = f(x)$ and a solution $x: [0,T] \rightarrow \R^n$, we say that this orbit can be simulated by a system $\frac{d}{dt} y = g(y)$ on $m$ variables if there is an \emph{affine map} $h:\R^n \rightarrow \R^m$ such that $y(t) = h(x(t))$ is a solution to $\frac{d}{dt} y = g(y)$. Recall that an affine map is a map of the type $h(x) = Bx + c$ for a matrix $B$ and a vector $c$. With this definition we can state our main result as follows:

\begin{theorem}\label{thm:simulating_linear_systems}
Every solution $x:[0,T] \rightarrow \R^n$ of a linear dynamical system  $\frac{d}{dt} x = Ax$ can be simulated by an ROS system.
\end{theorem}

The proof will be done in two steps. First we will show that every linear system can be simulated by a competitive linear system and then we will show that every competitive linear system can be simulated by a ROS system. Here, we say that an $n \times n$ matrix $A$ is \emph{competitive} if $A_{ij} \leq 0$ for $i \neq j$ (and it is \emph{purely competitive} if it is competitive and $A_{ii} = 0$). A linear dynamical system $\frac{d}{dt} x = Ax$ is (purely-)competitive if the matrix $A$ is (purely-)competitive.

\subsection{Competitive linear systems}

In this section, we will show that every linear dynamical system can be simulated by a purely competitive linear system. Our main tool will be the following lemma, essentially showing that given any matrix $A$, we can implement it as a submatrix of a purely competitive matrix $B$.

\begin{lemma}\label{lemma:competitive_matrix}
For any $n \times n$ matrix $A$, there is an $m \times m$ purely competitive matrix $B$ and an injective linear transformation $T:\R^n \rightarrow \R^m$ such that $TA = BT$.
\end{lemma}

We prove Lemma \ref{lemma:competitive_matrix} in Appendix \ref{appendix:competitive_matrix} by establishing sublemmas regarding the Jordan decomposition of non-negative matrices.  Lemma \ref{lem:nonnegative-matrix-eigenvalue} shows that every complex number can appear as the eigenvalue of a non-negative matrix. The in Lemma \ref{lem:nonnegative-matrix-jordan} we show that every possible Jordan block appears in the Jordan decomposition of some non-negative matrix. With those pieces, we can then establish the following:

\begin{lemma}\label{lemma:competitive_system}
Every solution $x:[0,T] \rightarrow \R^n$ of a linear dynamical system  $\frac{d}{dt} x = Ax$ can be simulated by a purely competitive linear system.
\end{lemma}

\begin{proof}
For the matrices $B$ and $T$ in Lemma \ref{lemma:competitive_matrix},
consider the purely-competitive linear system $\frac{d}{dt} y = By$ as well as the transformation $h(x) = Tx$. Now, if $x(t)$ is a solution to $\frac{d}{dt} x = Ax$  then define $y(t) = Tx(t)$ and observe that:
$y'(t) = T x'(t) = TA x(t) = BT x(t) = B y(t)$
hence $y(t)$ is a solution to $\frac{d}{dt} y = By$.
\end{proof}

\subsection{From competitive linear systems to ROS systems}

We now show how to construct any purely-competitive linear system within the ROS dynamics.

\begin{lemma}\label{lemma:competitive_to_ros}
Every solution $x(t):[0,T] \rightarrow \R^n$ of a purely-competitive linear dynamical system  $\frac{d}{dt} x = Ax$ can be simulated by an ROS system.
\end{lemma}

The proof of Lemma \ref{lemma:competitive_to_ros} first maps a trajectory $x(t)$ to a bounded box $[1.1, 1.9]^n$ using an affine transformation. Then it constructs an ROS system with $n+2$ bidders with the first $n$ bidders corresponding to the variables of the original system and  two auxiliary bidders are introduced to have constant bid multipliers throughout the dynamics. The role of those bidders will be to create price pressure on the remaining bidders in order to simulate a constant term in a linear system. Finally, we simulate linear terms in the dynamics by creating items that two bidders are interested in. If the $\frac{d x_i}{ dt}$ has as $-A_{ik} x_{k}$ term, the we create an item $j$ such that $i$ gets the item and bidder $k$ is the price setter.

The full details of the reduction can be found in Section \ref{appendix:competitive_to_ros}. In Section \ref{sec:example_linear_simulation} we give an end-to-end example of the reduction. In that example we provably construct a system such that the first two multipliers trace a circular orbit $m_1(t) = c+a \cos t$ and $m_2(t) = c + a \sin t$ for constants $a$ and $c$.

\section{Simulating Discrete Boolean Circuits And Networks}\label{sec:circuits}

In this section, we show that ROS systems can exhibit a different type of complex behavior: we show that they are able to simulate arbitrary boolean circuits. In fact, we will prove something a little more general and show that ROS systems can represent arbitrary \emph{boolean networks}. A boolean network of size $n$ is a collection of $n$ boolean variables $X_1, X_2, \dots, X_n$, each of which is associated with a constraint $X_i = f_i(X_1, X_2, \dots, X_n)$, for some associated Boolean function $f_i : \{0, 1\}^n \rightarrow \{0, 1\}$. For example, the following system is a boolean network of size 3:

\begin{equation}\label{eq:nor-system}
X = \NOR(Y, Z) \qquad Y = \NOR(Z, X) \qquad Z = \NOR(X, Y)
\end{equation}

An assignment of truth values ($0$ or $1$) to $X_i$ \emph{satisfies} the boolean network if each of the constraints $X_i = f_i(X)$ is satisfied. We will show that we can take any boolean network and produce an ROS system with the property that if the bid multipliers in the ROS system converge to a stable equilibrium, then there must exist a satisfying assignment to the boolean network. Moreover, the system will be set up in such a way that the multipliers of variable-bidders must converge to one of two values (called $\high$ and $\low$) corresponding to the two possible truth values, allowing us to read off the satisfying assignment from the ROS system equilibrium.

For example, given the network in \eqref{eq:nor-system}, we can use our reduction to construct an ROS system with $21$ bidders such that $3$ of these bidders correspond to the variables $X$, $Y$ and $Z$. These bidders are constructed so that the dynamics pulls their multipliers either towards $\high = 3$ or $\low=1.5$. For the $X$-bidder, for example, if both $Y$ and $Z$ are close to $\low$, the dynamics pulls the multipliers towards $\high$; otherwise $X$ is pulled towards $\low$. Note that this simulates the behavior of a $\NOR$ gate. The behavior of these three bidders is shown in Figure~\ref{fig:circuit_a}. The multipliers converge to $(\high, \low, \low)$, which corresponds to the assignment $(X, Y, Z) = (1, 0, 0)$, which is a satisfying assignment to the above network.

\begin{figure}[h]
\begin{center}
\includegraphics[width=\textwidth]{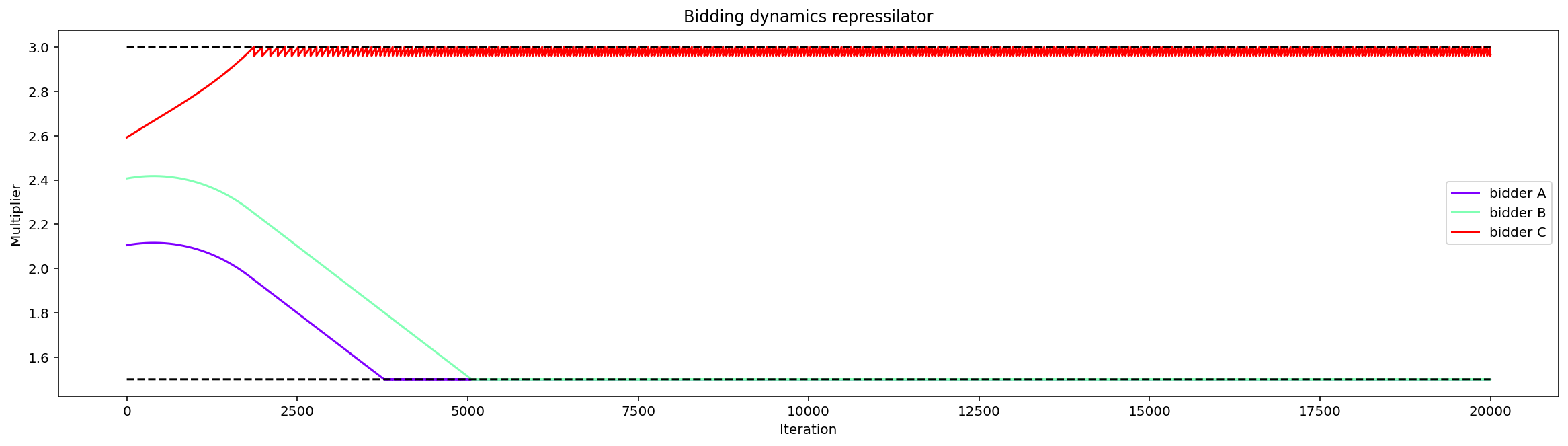}
\end{center}
\caption{Behavior of bid multipliers in an ROS system encoding boolean equations $X = \NOR(Y, Z),Y = \NOR(Z, X),Z = \NOR(X, Y)$.}
\label{fig:circuit_a}
\end{figure}

\begin{theorem}\label{thm:main_circuit}
Given any boolean network $\mathcal{C}$, we can construct an ROS dynamical system $\mathcal{S}$ such that every satisfying assignment of $\mathcal{C}$ corresponds to a unique stable\footnote{Here ``stable'' refers to a strong, coordinate-wise notion of stability (see Discussion in Appendix \ref{subsec:combining-gates}).} equilibrium of $\mathcal{S}$, and vice versa. Moreover, if it is possible to construct the $n$ boolean functions $f_i$ in $\mathcal{C}$ with a total of $G$ $\NOR$ gates, the system $\mathcal{S}$ will have $O(G)$ bidders.
\end{theorem}

The main idea behind the proof of Theorem \ref{thm:main_circuit} is to demonstrate that it is possible to construct ROS dynamics that simulate a NOR gate. We give high-level intuition for how to do this later in Section~\ref{subsec:nor-gate} (for a detailed proof, see Appendix~\ref{sec:nor-proof}).


Note that Theorem \ref{thm:main_circuit} does not imply that the ROS system is always guaranteed to converge when there are satisfying assignments (in fact, since the problem of finding a satisfying assignment to a Boolean network is NP-hard, it would be surprising if this were the case). However, in the special case where the Boolean network is \emph{acyclic} (each $f_i$ only depends on the variables from $X_1$ to $X_{i-1}$), it corresponds to the execution of a Boolean circuit, and we can show the associated ROS dynamics are guaranteed to converge to the unique stable equilibrium.

\begin{theorem}\label{thm:circuit-convergence}
If $\mathcal{C}$ is an acyclic boolean network, then for almost all initial conditions, the ROS system $\mathcal{S}$ converges to a stable equilibrium.
\end{theorem}

Finally, some boolean networks do not have any satisfying assignments. These systems will necessarily exhibit periodic (or other non-convergent) behavior\footnote{To be precise, such systems still can have unstable equilibria but the dynamics will not converge to these unless they start at one of these points.}. One example is the system:
$$X = \NOT(Y) \qquad Y = \NOT(Z) \qquad Z = \NOT(X)$$
We depict the behavior of multipliers of the corresponding ROS system in Figure \ref{fig:circuit_b}. In Section \ref{sec:clock} we show how to apply this phenomenon to construct a clock.

\begin{figure}[h]
\begin{center}
\includegraphics[width=\textwidth]{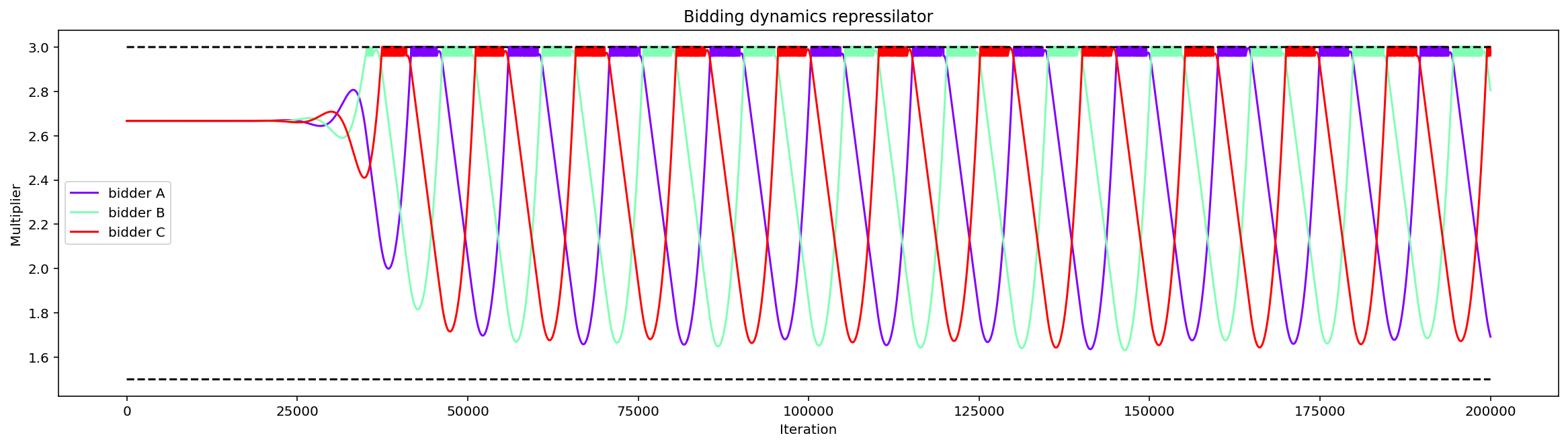}
\end{center}
\caption{Behavior of bid multipliers in an ROS system encoding boolean equations $X = \NOT(Y),Y = \NOT(Z), Z = \NOT(X)$.}
\label{fig:circuit_b}
\end{figure}

\cameraready{Details of the construction of the NOR gate are deferred to the full version of this paper.}{
\subsection{\NOR Gate Construction}\label{subsec:nor-gate}

In this section we give a high-level overview of our NOR gate construction (which drives the proofs of Theorems \ref{thm:main_circuit} and \ref{thm:circuit-convergence}). The key idea of the construction is illustrated in Figure~\ref{fig:nor-gate} in the appendix. Consider $k$ input bidders $x_1, \ldots, x_k$ ($k \geq 0$) and one output bidder $y$. There are two items in this system $\sH$ and $\sL$ -- all the bidders are interested in $\sL$, but only $y$ is interested in $\sH$. The two items are designed to have the following significance for the output bidder $y$:
\begin{itemize}
  \item $\sH$: A high (fixed) price item (with price larger than its value to $y$). $y$ is encouraged to win this item when competition for the other item is weak.
  \item $\sL$: A low (but dynamically) priced item (with price lower than its value to $y$). The price depends the competition pressure from the input bidders, but $y$ will always win this item: even with the highest level of competition, the price will still be lower than its value to $y$.
\end{itemize}
The values and prices of these items are chosen in a way such that
\begin{itemize}
  \item When none of the bidders $x_i$ are competing with $y$ on $\sL$ (i.e., all $x_i < \threshold$), $y$ will increase to $\high$ and win both items.
  \item When any of the $x_i$ are competing with $y$ on $\sL$ (i.e., any $x_i > \threshold$), $y$ cannot afford $\sH$ and will keep their multiplier at $\low$.
\end{itemize}

To realize the above idea, one needs to carefully select the parameters to restrict all bid multipliers within a reasonable range. In particular, it is necessary that:
\begin{itemize}
  \item $x_i$ never wins $\sL$ and $y$ always wins $\sL$;
  \item $y$ wins $\sH$ at \high while does not win at \low;
  \item If any of $x_i$ are above \threshold, the resulting price of $\sL$ is high enough to push $y$ back to \low.
\end{itemize}

We cover these construction details (and associated proof) in Appendix~\ref{sec:nor-proof}.
}

\clearpage
\newpage
\newpage

\bibliographystyle{apalike}
\bibliography{references,ms,refs}   

\appendix
\section{Missing proofs from Section \ref{sec:qualitative}}\label{appendix:proofs_qualitative}

\begin{proof}[Proof of Lemma \ref{lemma:ros:UI}]
    If we treat $x_{ij}$ and $p_{ij}$ as a function of bids in the second price auction, then we can write:
    $$U_i = \sum_j \E_{v_{ij}}[v_{ij} \hat x_{ij}(m_i v_{ij}; v_{ij}) - \hat p_{ij}(m_i v_{ij}; v_{ij})]$$
    where:
    $$\hat x_{ij}(b_{ij}; v_{ij}) = \E[ x_{ij}(m_1 v_{1j}, \hdots, m_{i-1} v_{i-1,j}, b_{ij}, m_{i+1} v_{i+1,j}, \hdots, m_n v_{nj}) \mid v_{ij}]$$
    and similarly for $\hat p_{ij}(b_{ij}; v_{ij})$. Because we are in the smooth limit, the functions $\hat x_{ij}$ and $\hat p_{ij}$ are $C^1$ and since it is a second price auction, they satisfy $$\frac{\partial \hat p_{ij}}{\partial b_{ij}} = b_{ij} \frac{\partial \hat x_{ij}}{\partial b_{ij}}$$ Therefore we have:

    $$\begin{aligned}
    \frac{\partial U_i}{\partial m_i} & = \sum_j \E_{v_{ij}}\left[v_{ij} v_{ij} \frac{\partial \hat x_{ij}}{\partial b_{ij}}(m_i v_{ij}; v_{ij}) - \frac{\partial \hat p_{ij}}{\partial b_{ij}}(m_i v_{ij}; v_{ij})\right] \\ &
    = \sum_j \E_{v_{ij}}\left[ \frac{\partial \hat x_{ij}}{\partial b_{ij}}(m_i v_{ij}; v_{ij}) v_{ij}^2 (1-m_i)\right] \end{aligned}$$
    Hence $\frac{\partial U_i}{\partial m_i} > 0$ for $m_i < 1$ and $\frac{\partial U_i}{\partial m_i} < 0$ for $m_i > 1$. Furthermore, since $U_i = 0$ for $m_i = 0$, we have $U_i \geq 0$ for $m_i \leq 1$.
\end{proof}

\begin{proof}[Proof of Lemma \ref{lemma:invariance_above_1}]
Consider any $m\in \R^n$ with $m_i = 1$.  Note that $U_i(m) \geq 0$ by Lemma \ref{lemma:ros:UI}.  Thus the vector field along this facet is nonnegative in the $i$ coordinate.  Therefore the region $[1,\infty)^n$ is positively invariant with respect to the associated flow $\phi$.
\end{proof}

\begin{figure}[h]
\begin{center}
\includegraphics[width=0.6\textwidth]{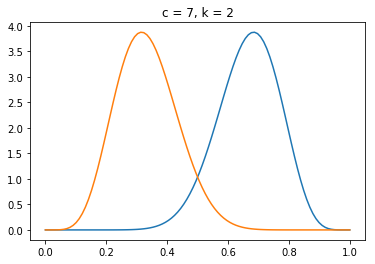}
\end{center}
\caption{Density of distributions $\text{Beta}(2c,c)$ (blue) and $\text{Beta}(c,2c)$ (orange) for $c=7$.}
\label{fig:beta_distribution}
\end{figure}

\begin{figure}[h]
\begin{center}
\hfill
\includegraphics[width=0.48\textwidth]{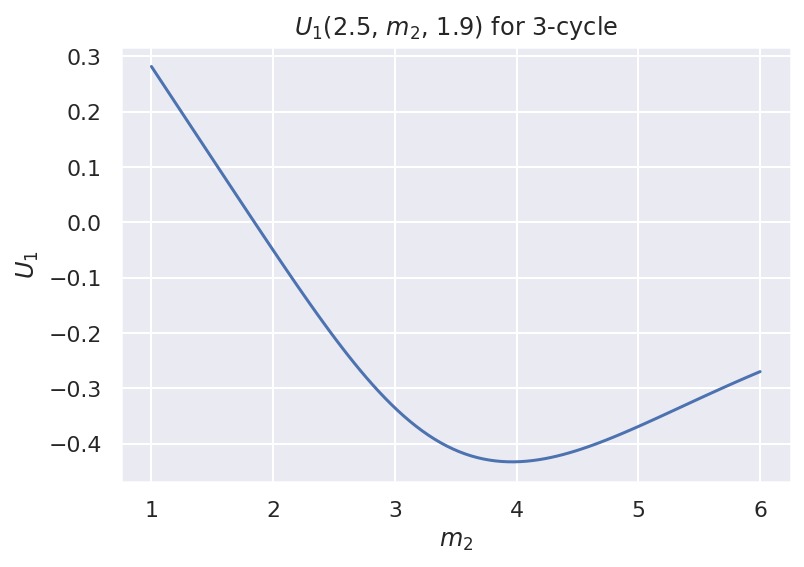}
\hfill
\includegraphics[width=0.48\textwidth]{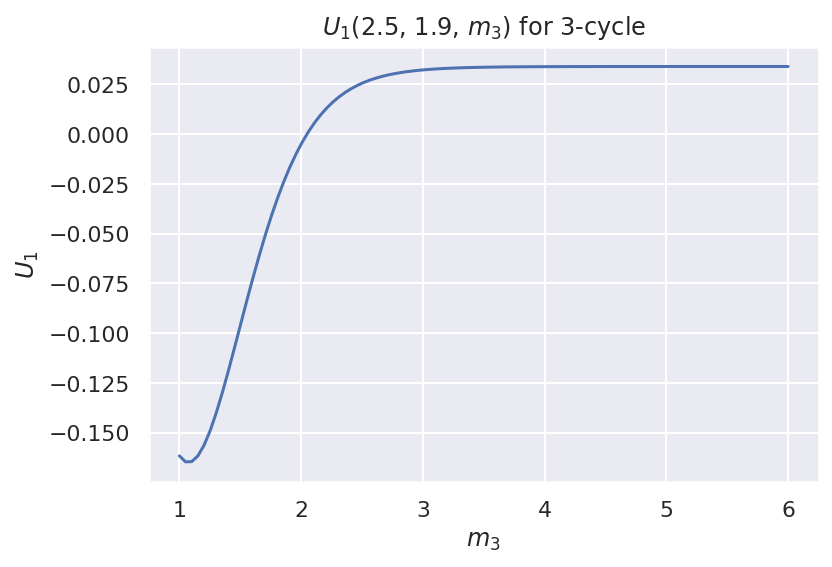}
\hfill
\end{center}
\caption{Plots of $U_1(2.5,x,1.9)$ [left] and $U_1(2.5,1.9,x)$ [right] for system Figure \ref{fig:repressilator_graphs} (b).}
\label{fig:U1_3cycle}
\end{figure}

\begin{figure}[h]
\begin{center}
\hfill
\includegraphics[width=0.48\textwidth]{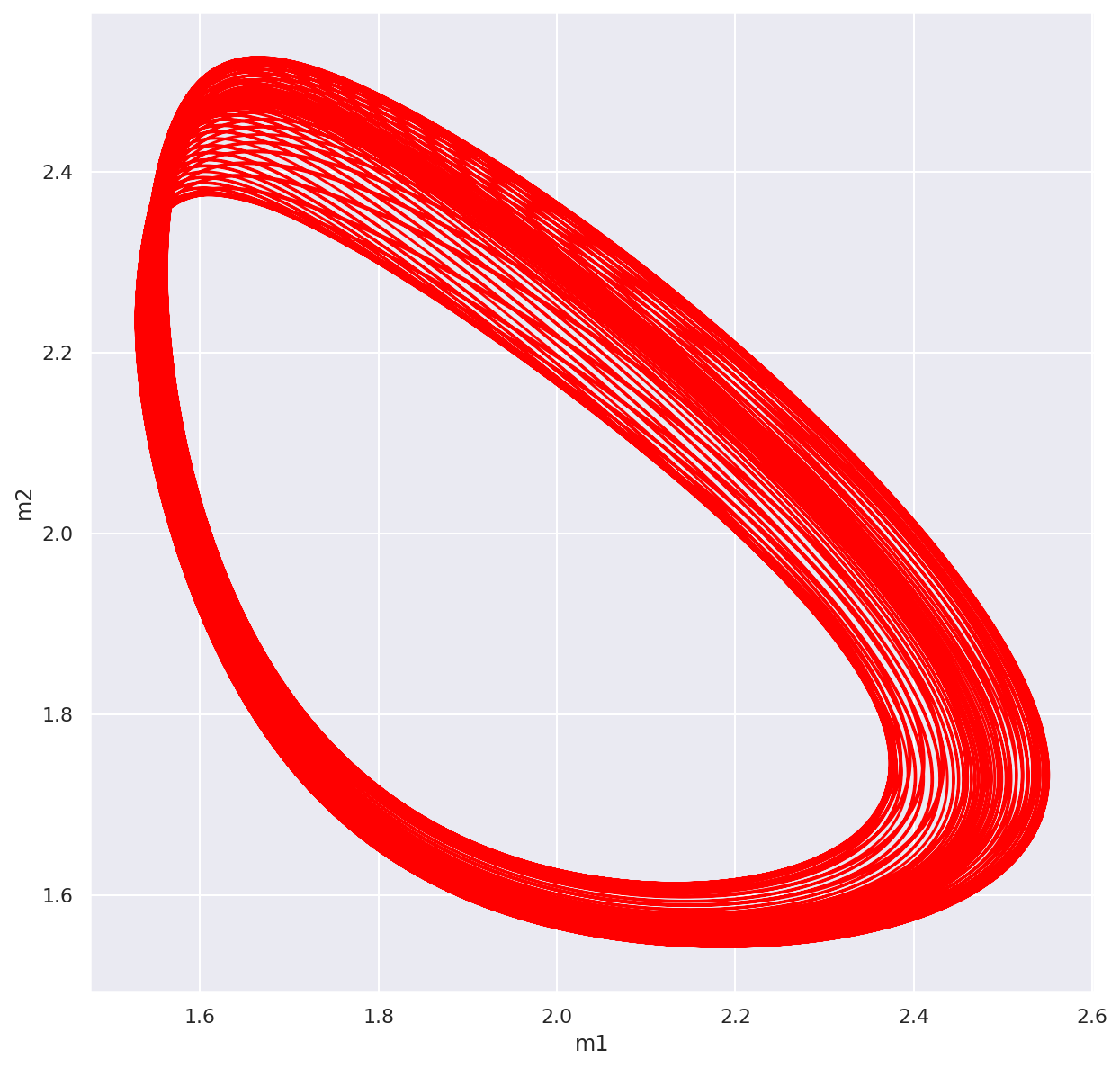}
\hfill
\includegraphics[width=0.48\textwidth]{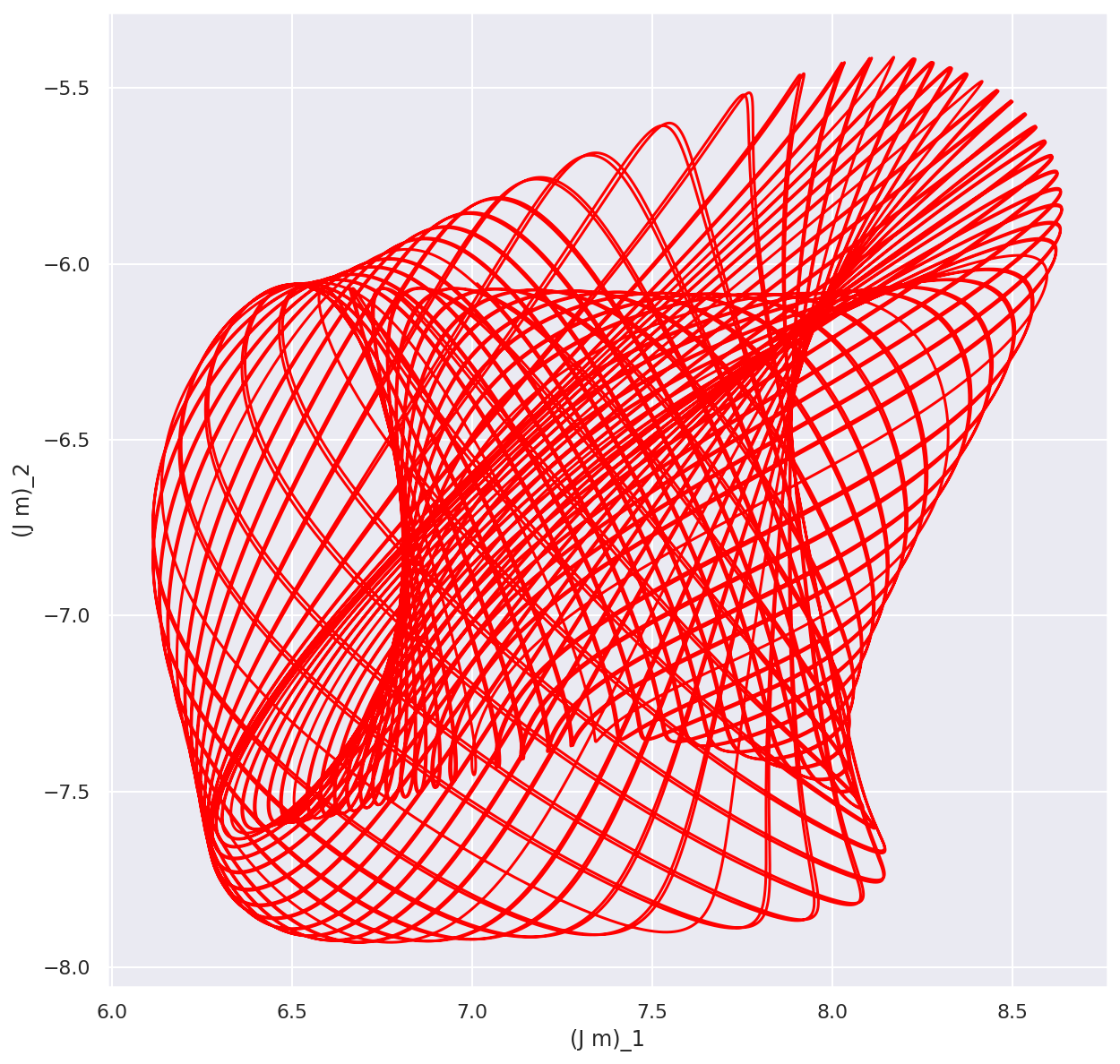}
\hfill
\end{center}
\caption{Quasi-periodic orbit corresponding to the first graph in Figure \ref{fig:repressilator_couplings}. The first graph corresponds to the orbit of multipliers $(m_1(t), m_2(t))$. The second graph is a random projection of the orbit $m(t)$.}
\label{fig:quasiperiodic_1}
\end{figure}

\begin{proof}[Proof of Lemma \ref{lemma:first_price_dynamics}] Let $\hat x_{ij}(b_{ij}; v_{ij})$ and $\hat p_{ij}(b_{ij}; v_{ij})$ be the expected allocation and payments of bidder $i$ for item $j$ when the bidding strategy of other agents are fixed. Following the notation in Lemma \ref{lemma:ros:UI} we can write the utility as $U_i = \sum_j \E[v_{ij} \hat x_{ij}(m_i v_{ij}; v_{ij}) - \hat p_{ij}(m_i v_{ij}; v_{ij})]$. Since under a first price auction we have: $\hat p_{ij}(b_{ij}; v_{ij}) = b_{ij} \hat x_{ij}(b_{ij}; v_{ij})$ it follows that $U_i = \sum_j \E[(1-m_i) v_{ij} \hat x_{ij}(m_i v_{ij}; v_{ij})]$. Hence $U_i \geq 0$ for $m_i \leq 1$ and $U_i \geq 1$ for $m_i \geq 1$. In the smooth limit, those inequalities are strict for $m_i \neq 1$ hence the dynamics must converge to $m_i = 1$.
\end{proof}

\section{Missing Proofs from Section \ref{sec:linear_system}}

\subsection{Proof of Lemma \ref{lemma:competitive_matrix}}\label{appendix:competitive_matrix}

We will focus on establishing Lemma \ref{lemma:competitive_matrix}. To do so, we will make use of the following two sublemmas regarding the Jordan decomposition of non-negative matrices. The first lemma shows that every complex number can appear as the eigenvalue of a non-negative matrix.

\begin{lemma}\label{lem:nonnegative-matrix-eigenvalue}
Given any $\lambda \in \C$, there exists a non-negative matrix $M$ containing $\lambda$ as an eigenvalue with multiplicity 1.
\end{lemma}
\begin{proof}
We reproduce an explicit construction from \cite{voneitzen_se}. We will restrict our attention to $\lambda$ with $\mathrm{Im}(\lambda) \geq 0$ (since if $\lambda$ is an eigenvalue of $M$, so is $\overline{\lambda}$). If $\lambda = a + bi$ (with $a, b \geq 0$), then:

$$
\begin{bmatrix}a&b&0&0\\0&a&b&0\\0&0&a&b\\b&0&0&a\end{bmatrix}\begin{bmatrix}1\\i\\-1\\-i\end{bmatrix}=(a+bi)
\begin{bmatrix}1\\i\\-1\\-i\end{bmatrix}$$

\noindent
Similarly, if $\lambda = -a + bi$, then:

$$\begin{bmatrix}0&b&a&0\\0&0&b&a\\a&0&0&b\\b&a&0&0\end{bmatrix}\begin{bmatrix}1\\i\\-1\\-i\end{bmatrix}=(-a+bi)
\begin{bmatrix}1\\i\\-1\\-i\end{bmatrix}.$$
\end{proof}

The second sublemma extends Lemma \ref{lem:nonnegative-matrix-eigenvalue} to show that every possible Jordan block appears in the Jordan decomposition of some non-negative matrix.

\begin{lemma}\label{lem:nonnegative-matrix-jordan}
Given any $\lambda \in \C$ and $d > 1$, there exists a non-negative matrix $M$ whose Jordan decomposition contains a Jordan block with eigenvalue $\lambda$ and multiplicity $d$.
\end{lemma}
\begin{proof}
By Lemma \ref{lem:nonnegative-matrix-eigenvalue}, there exists a non-negative matrix $M_{\lambda}$ which has $\lambda$ as an eigenvalue with multiplicity $1$. Let $J_d$ be the $d$-by-$d$ Jordan block with eigenvalue $1$ and multiplicity $d$ (note that $J_d$ is also a non-negative matrix). We will let $M$ equal the Kronecker product $M_{\lambda} \otimes J_d$. By Theorem 4.3.17 of \cite{horn1991topics}, the Jordan decomposition of $M$ will contain a Jordan block with eigenvalue $\lambda$ and multiplicity $d$.
\end{proof}

With Lemmas \ref{lem:nonnegative-matrix-eigenvalue} and \ref{lem:nonnegative-matrix-jordan}, we can show that we can construct purely competitive matrices with arbitrary Jordan decompositions.

\begin{lemma}\label{lem:purely-competitive-jordan}
Given any collection of Jordan blocks $\{J_1, J_2, \dots, J_k\}$ (each $J_i$ with some eigenvalue $\lambda_i$ and multiplicity $d_i$), there exists a purely competitive matrix $B$ whose Jordan decomposition contains this multiset of Jordan blocks.
\end{lemma}
\begin{proof}
We first construct a non-negative matrix $M$ whose Jordan decomposition contains this collection of Jordan blocks. To do so, we can apply Lemma \ref{lem:nonnegative-matrix-jordan} to obtain a non-negative matrix $M_i$ whose Jordan decomposition contains $J_i$, and let $M$ be the block-diagonal matrix formed by all the $M_i$ (i.e., the direct sum of these matrices). Since the Jordan decomposition of a direct sum of a sequence of matrices is the additive union of the Jordan decompositions of the matrices, $M$ contains this collection of Jordan blocks as a subset of its Jordan decomposition.

Now, $M$ is a non-negative matrix, not a purely competitive matrix. To fix this, we will let $B = M \otimes S$, where

$$S = \begin{bmatrix}0&-1\\-1&0\end{bmatrix}.$$

Note that this guarantees that $B$ is a purely competitive matrix (every entry of $B$ will be non-positive, and every diagonal element equals $0$). On the other hand, since $S$ has an eigenvector of eigenvalue $1$ (namely, $(1, -1)$), again by Theorem 4.3.17 of \cite{horn1991topics}, $B$ will also contain each of the Jordan blocks $J_i$ in its Jordan decomposition.
\end{proof}

We can now prove Lemma \ref{lemma:competitive_matrix}.

\begin{proof}[Proof of Lemma~\ref{lemma:competitive_matrix}]
Let $J$ be the Jordan normal form of $A$ (so $A = PJP^{-1}$). By Lemma \ref{lem:purely-competitive-jordan}, there exists a purely competitive $m \times m$ matrix $B$ whose Jordan normal form $J'$ contains all the Jordan blocks of $J$. Write $B = QJ'Q^{-1}$. Since $J'$ contains all the Jordan blocks of $J$, there exists a linear embedding $\Pi:\R^n \rightarrow \R^m$ with the property that $\Pi J = J'\Pi$.

Now, note that if we take $T = Q \Pi P^{-1}$, $T$ satisfies $TA = BT$. The real part $T_{r} = \mathrm{Re}(T)$ also then satisfies $T_{r}A = BT_{r}$, as desired.
\end{proof}

\subsection{Proof of Lemma \ref{lemma:competitive_to_ros}}\label{appendix:competitive_to_ros}

\begin{proof}[Proof of Lemma \ref{lemma:competitive_to_ros}]
Choose a constant $c \in \R$ and a vector $b \in \R^n$ such that $y(t) = b + c \cdot x(t) \in (1.1, 1.9)^n$. Observe that $y(t)$ satisfy the equation $\frac{d}{dt} y = A(y-b)$ since:
$$\frac{d}{dt} y = c \frac{d}{dt} x = c Ax = c A \left(\frac{y-b}{c}\right) = A (y-b)$$
Now, recall that $A$ is purely competitive, so $A_{ii} = 0$ and $A_{ij} \leq 0$. We will now construct an ROS system on $n+2$ variables such that the orbit projected on the first $n$ variables is exactly $y(t)$. To do so, we will construct gadgets to simulate each term of the system $\frac{d}{dt} y = A(y-b)$. We need gadgets that simulate constant terms $b_i$ as well as linear terms with negative coefficient $A_{ij} y_j$. We will start with an auxiliary gadget that simulates a bidder with a constant multiplier.\\

\noindent \emph{Bidders:} For each variable $y_i$ of the original system we will construct a bidder $i$ in the ROS system. The initial setting of multipliers will be $m_i(0) = y_i(0) \in [1.1, 1.9]$. We will also construct two additional bidders which we will call $A$ and $B$ with initial multipliers $m_A(0) = m_B(0) = 2$.

For the two auxiliary bidders we will create two items $a,b$ such that the values $v_{Aa} = v_{Bb} = 2$, $v_{Ab} = v_{Ba} = 1$. The remainder of the construction will ensure that $A$ only wins item $a$ and $B$ only wins item $b$ and the multipliers $m_A(t) = m_B(t) = 2$ for $t \in [0,T]$.\\

\noindent \emph{Constant terms $b_i$:} We will now construct items that simulate the constant term $b_i$ in the original system. Let's start by assuming that $b_i < 0$.

First we show how to create a gadget such that bidder $i$ has constant utility $-0.1$.  We add item $j$ that that only agent $i$ and agent $A$ are interested in. Their values for which $v_{ij} = 1.9$, $v_{Aj} = 1$. For $m_A=2$ and $m_i \in [1.1, 1.9]$, $i$ wins the item and the utility it derives from it is $1.9 - 2 = -0.1$.

This gadget can be used to create any negative utility as follows: if $\abs{b_i} > 0.1$ the create $k = \lfloor -b_i / 0.1 \rfloor $ copies of the gadget. For the remaining value $w = -b_i + k \cdot 0.1$, create an extra copy of the gadget with the values scaled by $-w / 0.1$.

For the case $b_i > 0$ we can use the same gadget but changing $v_{ij} = 2.1$. Again $i$ will win the item and obtain constant utility $2.1 - 2.0 = 0.1$. The same argument applies to go from $0.1$ to any constant.\\

\noindent \emph{Negative linear term $A_{ik} y_k$:} Create one item $j$ such that $v_{kj} = \abs{A_{ik}}$ and $v_{ij} = 2 \abs{A_{ik}}$. Now, for any multipliers $m_i, m_k \in [1.1, 1.9]$, bidder $i$ wins the item and gets utility $2 \abs{A_{ik}} - \abs{A_{ik}} m_k = c + A_{ik} m_k$ for a constant $c$. Combining it with the constant utility gadget above, we can remove the constant term and simply get $A_{ik} m_k$.\\

\noindent \emph{Putting it all together:} Using the gadgets above for each of the terms, we can construct an ROS system $\frac{d}{dt} m_i = U_i(m)$ with $n+2$ bidders that has the following properties when: $m \in D = (1.1, 1.9)^n \times (2-\epsilon, 2+\epsilon)^2$:
\begin{itemize}
\item the auxiliary bidders $A$ and $B$ only win items $a$ and $b$ respectively, so $U_A(m) = U_B(m) = 0$.
\item bidder $i$ has utility $U_i(m) = A(m_{1..n}-b)$ where $m_{1..n}$ corresponds to the vector $m \in \R^{n+2}$ restricted to the first $n$ components.
\end{itemize}

In particular, all the utilities $U_i$ are $C^1$ in $D$. Therefore a solution exists and is unique. For that reason, we observe that $m(t) = (y(t), 2, 2)$ is a solution to the equation $\frac{d}{dt} m_i = U_i(m)$. Since the $U_i$ are $C^1$ in $D$, this is the unique solution. Finally, observe that $x(t)$ and $m(t)$ are related by the affine map $h(x) = (b + cx, 2, 2)$, i.e., $h(x(t)) = m(t)$.

\end{proof}

\subsection{Example of the reduction}\label{sec:example_linear_simulation}

We finally present an end-to-end example. Consider the linear system of equations:
$$\frac{dx}{dt}  = \begin{bmatrix}
    0 & -1 \\
    1 & 0
\end{bmatrix} x$$

The solution to this system is the periodic orbit $x(t) = (\cos t, \sin t)$ -- in this case, the orbit is the unit circle. The system is not competitive, but it can be simulated by the following purely-competitive system:

$$\frac{dz}{dt}  = \begin{bmatrix}
    0 & 0 & 0 & -1 \\
    -1 & 0 & 0 & 0 \\
    0 & -1 & 0 & 0 \\
    0 & 0 & -1 & 0
\end{bmatrix} z$$
whose solution is: $z(t) = (-\cos t, -\sin t, \cos t, \sin t)$. Let $A$ be the $4 \times 4$ matrix above. The affine mapping $h(x) = (-x, x)$. Now, consider the mapping $y = 1.5 + 0.4 z$ such that the dynamics is in $[1.1, 1.9]^4$. The system $\frac{d}{dt} y = Ay + 1.5 \cdot \mathbf{1}$ can be simulated by the following ROS system with $6$ bidders and $26$ items. The rows in the following table correspond to bidders and the columns to items. The first row specifies how many copies we need for each group of items. The first column correspond to the index of bidders. The remainder of the table corresponds to values $v_{ij}$. In Figure \ref{fig:circle} we show the dynamics of this ROS system with initial multipliers $m(0) = (1.1, 1.5, 1.9, 1.5, 2, 2)$. As expected, the orbit of the first two multipliers corresponds to an affine transformation of the unit circle (Figure \ref{fig:circle}).

\begin{center}
\begin{tabular}{c | c c c c | c c c c | c c }
 copies & 1x & & & & 5x & & & & 1x & \\
 \hline
 1 & 2 & 1 & & & 1.9 & & & &  & \\
 2 &  & 2 & 1 & &  & 1.9 & & &  & \\
 3 & & & 2 & 1 & & & 1.9 & &   & \\
 4 & 1 & & & 2 &   & & & 1.9 &   & \\
 A &  &  & & & 1 & 1 & 1 & 1 & 2 & 1 \\
 B &  &  & & &  & & & & 1 & 2 \\
 \hline
\end{tabular}
\end{center}

\begin{figure}[h]
\begin{center}
\hfill
\includegraphics[height=0.22\textwidth]{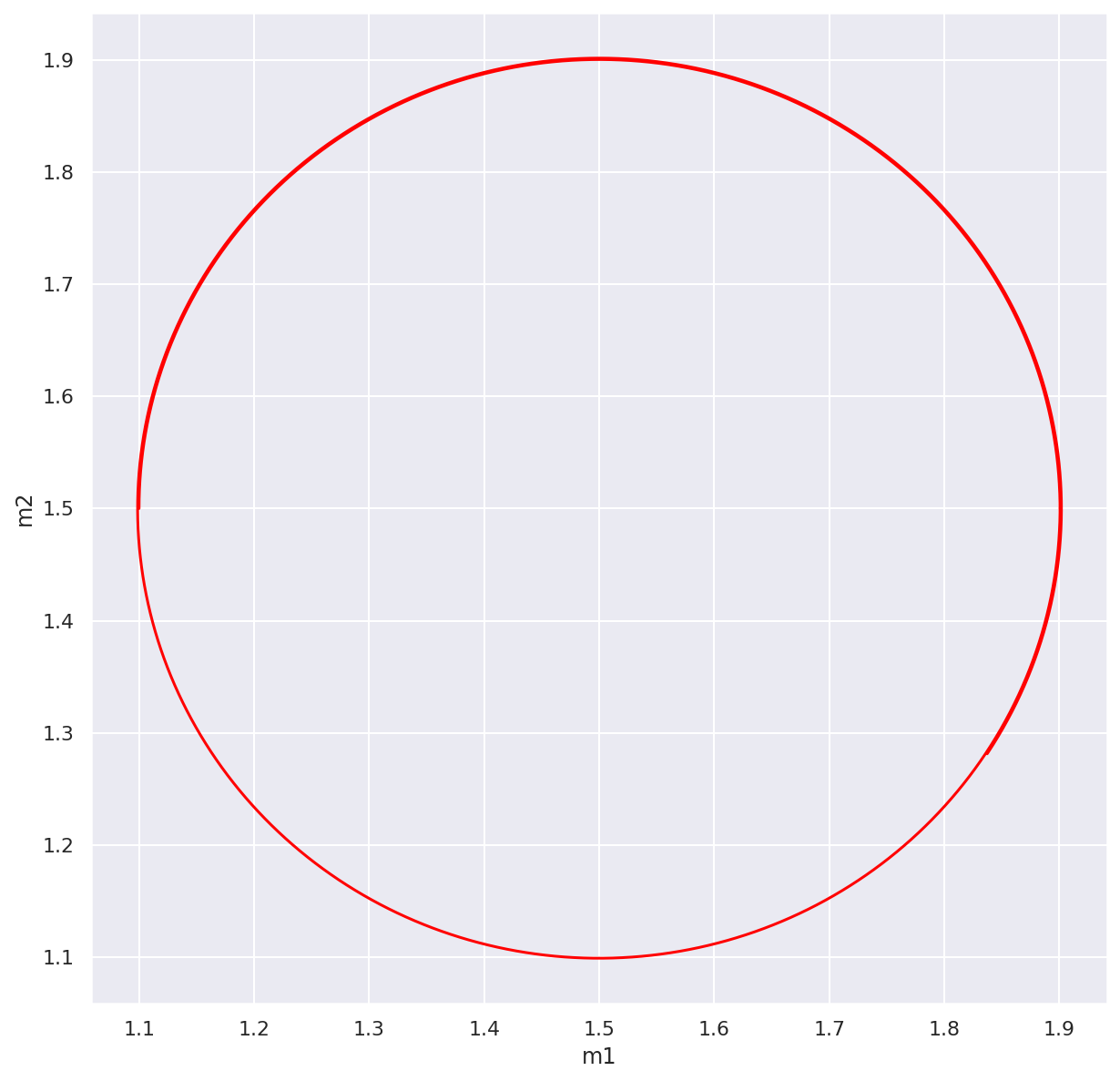}
\hfill
\includegraphics[height=0.22\textwidth]{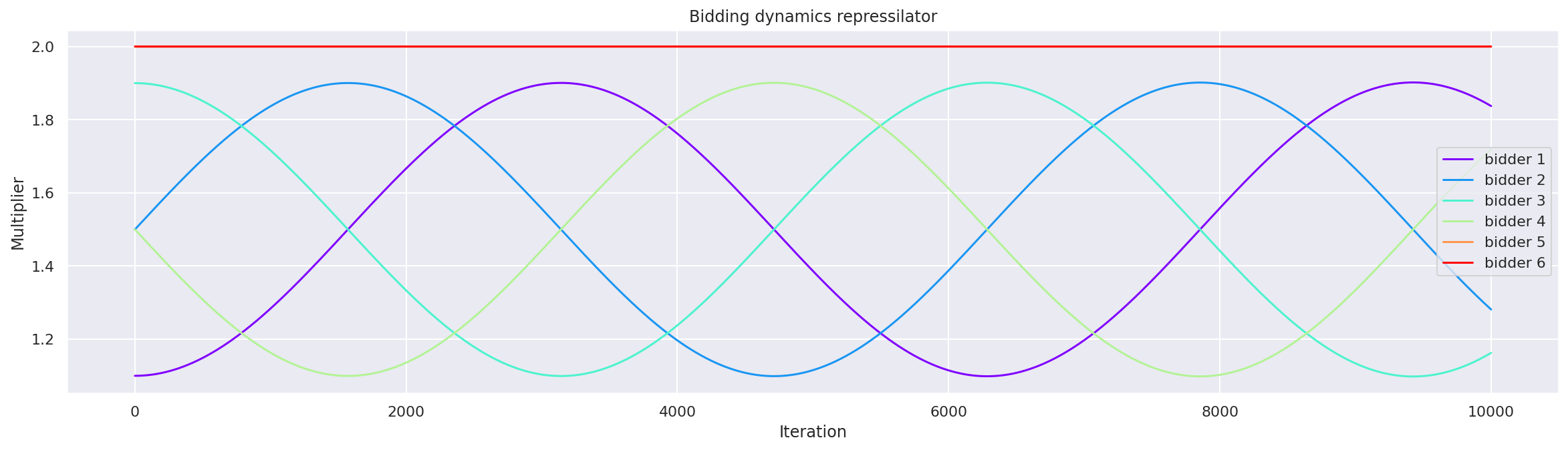}
\hfill
\end{center}
\caption{Behavior of the ROS system constructed in section \ref{sec:example_linear_simulation} simulating the linear system $x'_1 = -x_2, x'_2 = x_1$. The left plot shows the orbit $(m_1(t), m_2(t))$. The right plot shows how the multipliers of all bidders change over time.}
\label{fig:circle}
\end{figure}

\cameraready{}{
\section{Detailed Construction Of A NOR Gate}\label{sec:nor-proof}

\begin{figure}
  \centering
  \includegraphics[width=\textwidth]{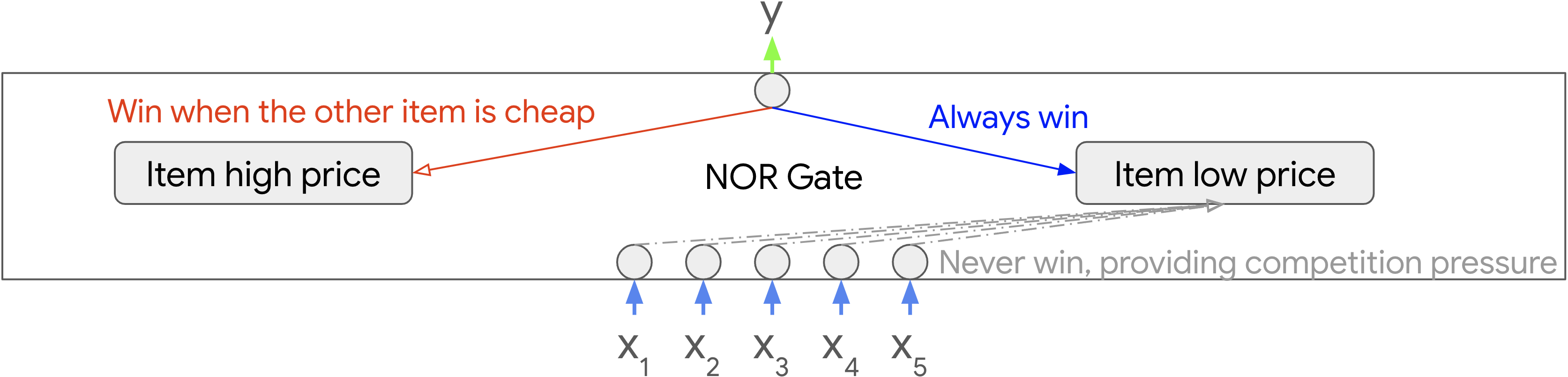}
  \caption{The structure of a \NOR gate construction with input variables $X_i = \cI[x_i > \mathsf{threshold}]$ and output variable $Y = \cI[y > \mathsf{threshold}]$. The small circles represent auto-bidding agents, while $x_i$ and $y$ are the corresponding bid multipliers. The gate functions as: $Y = \NOR(X_1, X_2, X_3, X_4, X_5) = \neg (X_1 \vee X_2 \vee X_3 \vee X_4 \vee X_5)$.}
  \label{fig:nor-gate}
\end{figure}

In this Appendix, we'll expand upon the intuition described in Section \ref{subsec:nor-gate} and give a rigorous construction of a NOR gate within ROS dynamics (from which it will be easy to prove Theorems \ref{thm:main_circuit} and \ref{thm:convergence}). Formally, our goal is to prove the following lemma.

\begin{lemma}[\NOR Gate]\label{lemma:nor-gate}
  For any $k \geq 0$ and $\threshold \in (\low, \high)$, there exists an ROS system with $\Theta(k)$ bidders and $\Theta(k)$ items, including $k$ ``input'' bidders $x_1, \ldots, x_k$ and one ``output'' bidder $y$, such that
  \begin{itemize}
    \item The gradient $\partial x_i/\partial t$ of any input bidder is $0$.
    \item When at least one input bidder $x_i > \threshold$, $y$ moves towards $\low$ (i.e., $U_y \coloneqq \partial y / \partial t < 0$ when $y > \low$ and $U_y > 0$ when $y < \low$);
    \item When all input bidders $x_i < \threshold$, then $y$ moves towards $\high$ (i.e., $U_y < 0$ when $y > \high$ and $U_y > 0$ when $y < \high$);
    \item When the maximum of the input bidders $\max_i x_i = \threshold$, then $y$ is converging towards the interval $[\low, \high]$ (i.e., $U_y < 0$ when $y > \high$, $U_y = 0$ when $y \in [\low, \high]$, $U_y > 0$ when $y < \low$).
  \end{itemize}

  Moreover, if the $x_i$ all satisfy $|x_i - \threshold| > \varepsilon$ for all time (for any fixed $\varepsilon > 0$), then $y$ will converge to $\low$ if all $x_i < \threshold$ and $y$ will converge to $\high$ otherwise.
\end{lemma}

We will do this in two steps. First, in Appendix \ref{subsec:simplified-nor-construction}, we will give a proof of Lemma \ref{lemma:nor-gate} under a slightly more flexible form of ROS dynamics, allowing ourselves to impose reserve prices and set floors and ceilings on individual multipliers. In Appendix \ref{subsec:reserve-ceiling} we will show how to relax these assumptions by implementing these features with additional bidders. Finally, in Appendix \ref{subsec:combining-gates} we will show how we can use Lemma \ref{lemma:nor-gate} to prove the two theorems of Section \ref{sec:circuits}.

\subsection{Simplified Construction with Reserve, Floor, and Ceiling}\label{subsec:simplified-nor-construction}

We first show a simplified construction of a \NOR gate allowing one to set reserve prices for items as well as floors and ceilings on the bid multipliers. For a complete construction without additional assumptions, we show in Appendix~\ref{subsec:reserve-ceiling} how one can implement reserve prices, floors, and ceilings through auxiliary bidders and items.

Consider a \NOR gate with $k$ input bidders $x_1,\ldots, x_k$ and one output bidder $y$. By abuse of the notations, each $x_i$ and $y$ will refer to the bidder and the bidder's bid multiplier. With the common floor and ceiling, all the bid multipliers are restricted within the range of $[\mfloor, \mceil]$. The valuation of each bidder on these two items are given by Table~\ref{tab:nor-simple}, where $V < T < C$ are the parameters to be determined later.

\begin{table}[h]
  \centering
  \begin{tabular}{c|c|c|c}
    items & value to $x_i$ & value to $y$ & reserve price  \\
    \hline
    $\sL$ & V & C & 0  \\
    \hline
    $\sH$ & 0   & T & C
  \end{tabular}
  \caption{The valuation and reserve prices for bidders on item $\sL$ and $\sH$.}
  \label{tab:nor-simple}
\end{table}

With the structure above, the boolean variables are represented as $X_i = \cI[x_i
> \mathsf{threshold}]$ and $Y = \cI[y > \mathsf{threshold}]$. We remain to figure out the quantitative relationships among those parameters such that (i) any of these bid multipliers only stabilizes at either \low or \high ($\mfloor \leq \low \leq \mathsf{threshold} < \high \leq \mceil$; (ii) the system only stabilizes with $Y = \neg \bigvee_{i=1}^k X_k$. Intuitively, we would like bidder $y$ to behave as:
\begin{itemize}
  \item When all $x_i = \low$, it is better off for $y$ to increase to \high to win both $\sL$ and $\sH$, which means that the quasi-linear utility when winning both items is non-negative:
  \begin{equation}\label{eq:nor-true}
    U_y(\high) = C - \low \cdot V + T - C \geq 0;
  \end{equation}
  \item When any of $x_i > \mathsf{threshold}$, the quasi-linear utility of $y$ to win both items becomes negative, while winning $\sL$ only remains non-negative:
  \begin{equation}\label{eq:nor-false}
    \left\{\begin{array}{l}
      U'_y(\low) = C - \max_i x_i \cdot V \geq C - \mceil \cdot V \geq 0  \\
      U'_y(\high) = C - \max_i x_i \cdot V + T - C < C - \mathsf{threshold} \cdot V + T - C \leq 0
    \end{array}\right..
  \end{equation}
\end{itemize}
In addition to the above, we also require the following to make sure that $y$ always wins $\sL$:
\begin{equation}\label{eq:nor-winl}
  b_y \geq \mfloor \cdot C > \mceil \cdot V \geq b_{x_i},
\end{equation}
and $y$ will win $\sH$ if and only if $y > \low$ (tie-breaking towards not allocating to $y$):
\begin{equation}\label{eq:nor-winh}
  \low \cdot T = C.
\end{equation}

Putting \eqref{eq:nor-true}, \eqref{eq:nor-false}, \eqref{eq:nor-winl}, \eqref{eq:nor-winh} together, we can prove Lemma \ref{lemma:nor-gate} (under these generalized ROS dynamics).

\begin{proof}[Proof of Lemma~\ref{lemma:nor-gate}]
We will show that if we set parameters in the above construction as:
  \begin{itemize}
    \item $1 \leq \mfloor < \low \leq \mathsf{threshold} < \high = \mceil < \low \cdot \mathsf{threshold}$;
    \item $T = \mathsf{threshold} \cdot V$ and $C = \low \cdot \mathsf{threshold} \cdot V$,
  \end{itemize}

\item
then $U_y$ in the resulting dynamics satisfies the properties we require of it in Lemma \ref{lemma:nor-gate}.

  To do so, it suffices to verify \eqref{eq:nor-true}, \eqref{eq:nor-false}, \eqref{eq:nor-winl}, \eqref{eq:nor-winh} one by one.

  \begin{itemize}
    \item \eqref{eq:nor-true}: $C - \low \cdot V + T - C = T - \low \cdot V = \threshold \cdot V - \low \cdot V \geq 0$;
    \item \eqref{eq:nor-false}: $C - \mceil \cdot V = \low \cdot \threshold \cdot V - \mceil \cdot V > 0$, and $C - \threshold \cdot V + T - C = T - \threshold \cdot V = \threshold \cdot V - \threshold \cdot V = 0$;
    \item \eqref{eq:nor-winl}: $\mfloor \cdot C = \mfloor \cdot \low \cdot \threshold \cdot V > \mfloor \cdot \mceil \cdot V \geq \mceil \cdot V$;
    \item \eqref{eq:nor-winh}: $\low \cdot T = \low \cdot \threshold \cdot V = C$.
  \end{itemize}

\begin{figure}
  \centering
  \includegraphics[width=0.5\textwidth]{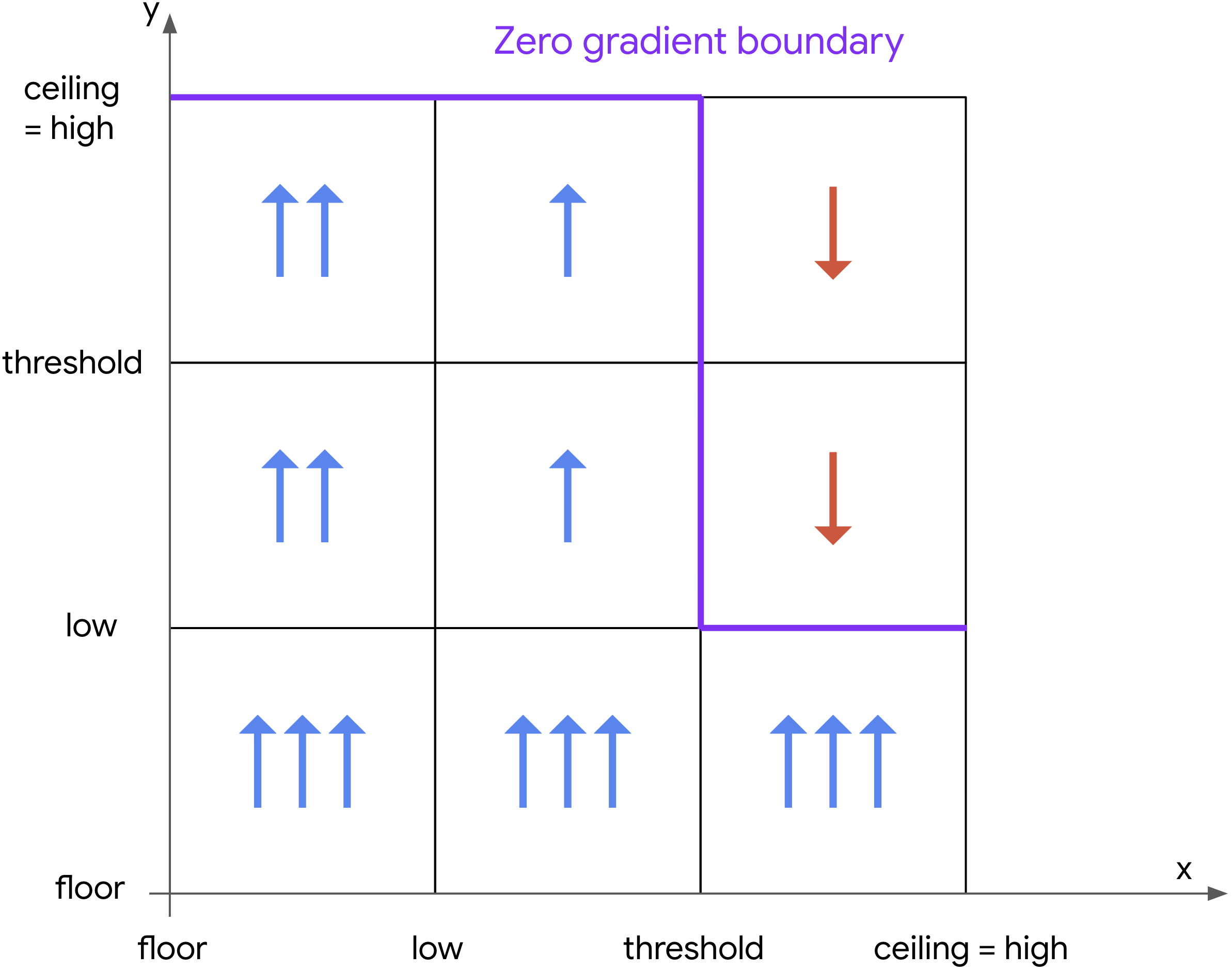}
  \caption{The bid multiplier update gradients of $y$ on different values of $y$ and $x = \max_i x_i$. The zero gradient area is highlighted as purple lines: (i) When $x < \mathsf{threshold}$, the only stable point is $y = \high = \mceil$; (ii) When $x > \mathsf{threshold}$, the only stable point is $y = \low$; (iii) When $x = \mathsf{threshold}$, $y$ has zero gradient in $[\low, \high]$. In this case, $y$ is not stable. Furthermore, $x_i = \mathsf{threshold}$ in general will not be stable when it is the output signal of another NOR-gate within the constructed system.}
  \label{fig:y-gradient}
\end{figure}

In particular, as a consequence of these relations, we have the following properties the gradient $U_y$ of $y$ (also depicted in Figure~\ref{fig:y-gradient}).
\begin{itemize}
  \item $U_y(\cdot) > 0$ whenever $x = \max_i x_i < \mathsf{threshold}$;
  \item For any $x = \max_i x_i > \mathsf{threshold}$, when $y \leq \low$, $U_y(\cdot) > 0$, otherwise $U_y(\cdot) < 0$;
  \item When $x = \max_i x_i = \mathsf{threshold}$, when $y \leq \low$, $U_y(\cdot) > 0$, otherwise $U_y(\cdot) = 0$.
\end{itemize}

Finally, if the $x_i$ are all well-separated from the threshold, these constraints on the gradient of $y$ imply that that $y$ will converge to the corresponding element of $\{\low, \high\}$.


\end{proof}

\subsection{Implementation of Reserve, Floor, and Ceiling}\label{subsec:reserve-ceiling}

In addition to the construction shown above with the assumption on reserve, floor, and ceiling, we now show how we can implement them by adding auxiliary bidders and items. In particular, the constructed auxiliary bidders will not react to the value of bidders not in the same gadget. Hence for any status of the system, all auxiliary bidders can converge to the designed equilibrium first, independent of the status of the non-auxiliary bidders. After that, the constructed functions of reserve, floor, and ceiling will work properly so that the non-auxiliary bidders can response or converge as expected in the simplified construction.

\paragraph{Reserve}
To implement an reserve $R$ on an item $\sA$ for all other bidders, we introduce two auxiliary bidders $\mathsf{aux}_1$, $\mathsf{aux}_2$ and one auxiliary item $\sU$ (see Table~\ref{tab:nor-reserve}). Without loss of generality, we assume that bidders never use bid multipliers smaller than $1$. Then clearly both of $\mathsf{aux}_1$ and $\mathsf{aux}_2$ must converge to $1$, and hence setting a competing price $R$ on item $\sA$ for all the rest bidders as expected.

\begin{table}[h]
  \centering
  \begin{tabular}{c|c|c|c}
    items & value to $\mathsf{aux}_1$ & value to $\mathsf{aux}_2$ & value to all other bidders \\
    \hline
    $\sU$ & $1$ & $1$ & $0$  \\
    \hline
    $\sA$ & $R$ & $R$ & specified by the construction
  \end{tabular}
  \caption{Two auxiliary bidders $\mathsf{aux}_1$, $\mathsf{aux}_2$ implements a reserve $R$ on item $\sA$, where both of their bid multipliers converge to $1$.}
  \label{tab:nor-reserve}
\end{table}

\paragraph{Ceiling} To implement the \mceil for each of the non-auxiliary bidders, we introduce one auxiliary item $\sC$ with reserve price $\mceil \cdot M$, where $M$ is a sufficiently large number such that $(\mceil - 1) \cdot M$ is larger than the total value of all other items to this bidder. Hence once the bidder wins $\sC$, its utility must become negative. Therefore its bid multiplier will be pushed down to no more than \mceil. To make \mceil an achievable bound, we break ties in favor to not allocate $\sC$ to this bidder.

\paragraph{Floor}
To implement the \mfloor for each non-auxiliary bidder $x$, we introduce two auxiliary items $\sE$ and $\sF$ (tie-breaks in favor of bidder $x$). Note that in the entire construction, we never introduce items with reserve-to-value ratio in $(1, \mfloor)$ for non-auxiliary items, so the quasi-linear utility of the bidder $x$ is guaranteed to be positive when $x < \mfloor$.
\begin{table}[h]
  \centering
  \begin{tabular}{c|c|c|c}
    items & value to $x$ & reserve & value to all other bidders \\
    \hline
    $\sE$ & $1$ & \mfloor & $0$  \\
    \hline
    $\sF$ & $\mfloor - 1$ & $0$ & $0$
  \end{tabular}
  \caption{Two auxiliary items implement the \mfloor for bidder $x$, which ensure that the quasi-linear utility of bidder $x$ is always positive when $x < \mfloor$.}
  \label{tab:nor-floor}
\end{table}

\subsection{Combining gates}\label{subsec:combining-gates}

Finally, we will discuss how to combine these NOR gates into larger boolean networks and circuits (and in the course of doing so, prove formal versions of Theorems \ref{thm:main_circuit} and \ref{thm:circuit-convergence}).

We begin with a note on the meaning of a stable equilibrium. We study a strong form of stability that we call \emph{coordinate-wise stability}, which we define as follows.

\begin{definition}
For a given instantiation of ROS dynamics, we say an equilibrium $m^* = (m^*_1, \ldots, m^*_n)$ is {\em coordinate-wise stable}, if for almost all $m \in B(m^*, \epsilon) \cap [1, \infty)^n$, for all $i \in [n]$ we have:

\begin{align*}
    \frac{dm_i}{dt} > 0 &\text{ if } m_i < m^*_i \\
    \frac{dm_i}{dt} < 0 &\text{ if } m_i > m^*_i.
\end{align*}
\end{definition}

In words, coordinate-wise stability means that upon a sufficiently small perturbation, each individual bidder's multiplier starts moving towards the equilibrium at a non-zero rate. Note that coordinate-wise stability implies the usual notion of local stability (which implies that the trajectory from $m'$ almost surely converges to $m$). For circuits (Theorem \ref{thm:circuit-convergence}) our results work with both coordinate-wise and local stability.

Now, we first show that we can restrict our attention to networks and circuits entirely comprised of 2-input $\NOR$ gates. Specifically, we say a Boolean network is a $\NOR$-network if each $f_i$ is of the form $\NOR(X_{j_1}, X_{j_2})$ for two indices $j_1, j_2 \in [n]$. Similarly, we call an acyclic $\NOR$-network a $\NOR$-circuit. We have the following structural result.

\begin{lemma}\label{lemma:universal-nor}
Given a boolean network $\mathcal{C}$, there exists a $\NOR$-network $\mathcal{C'}$ such that the satisfying assignments of $\mathcal{C}$ and $\mathcal{C'}$ are in bijective correspondence (and moreover, a satisfying assignment $X$ to $\mathcal{C}$ can be recovered from an assignment $X'$ to $\mathcal{C'}$ by a looking at a set of $|\mathcal{C}|$ of the coordinates of $\mathcal{C}$). The same equivalence exists between boolean circuits and $\NOR$-circuits.
\end{lemma}
\begin{proof}
This follows straightforwardly from the fact that every boolean function can be written as the composition of $\NOR$ gates. We simply introduce an additional boolean variable for the output of each intermediate gate.
\end{proof}

Given Lemma \ref{lemma:universal-nor}, it is sufficient (for both theorems) to work entirely with $\NOR$-networks. We will now describe how to use the single $\NOR$ gate construction of the previous sections to generate a construction for a full $\NOR$ network $\mathcal{C}$. The construction is very natural: for every variable $X_i$ in $\mathcal{C}$, we instantiate a bidder $x_i$ in $\mathcal{S}$. Then, for each of the NOR constraints $X_{i} = \NOR(X_{i_1}, X_{i_2})$, we embed a copy of the \NOR gate construction for these three bidders, letting $x_{i_1}$ and $x_{i_2}$ be the input bidders and $x_i$ be the output bidder (and adding to $\mathcal{S}$ the $O(1)$ items and auxiliary bidders we need for the construction to work).

We prove that this construction satisfies the required properties of Theorem \ref{thm:main_circuit}.

\begin{proof}[Proof of Theorem~\ref{thm:main_circuit}]
We characterize the stable equilibria of $\mathcal{S}$. Note that each bidder $x_i$ is the output bidder of exactly one of our embedded NOR gates and the input bidder of some number of NOR gates. Since being the input bidder of a NOR gate construction does not impact the bidder's gradient (first point of Lemma \ref{lemma:nor-gate}), $\partial x_i/\partial t$ is entirely determined by the values of $x_{i_1}$ and $x_{i_2}$ as described in Lemma \ref{lemma:nor-gate}.

Therefore, in order for $\partial x_i/\partial t$ to equal zero, one of the following three cases must be true:

\begin{itemize}
    \item $x_i = \high$, and $x_{i_1}, x_{i_2} < \threshold$.
    \item $x_i = \low$, and either $x_{i_1} > \threshold$ or $x_{i_2} > \threshold$.
    \item $\max(x_{i_1}, x_{i_2}) = \threshold$.
\end{itemize}

We first examine equilibria where the third bullet never holds. In such equilibria, each $x_i \in \{\low, \high\}$, and furthermore must correspond to a satisfying assignment of the $\NOR$-network (mapping $\low$ to $0$ and $\high$ to $1$). Moreover, note that this equilibrium is coordinate-wise stable (perturbations that keep each $x_i$ on the same side of $\threshold$ will converge to the same equilibrium). This shows every satisfying assignment corresponds clearly to a stable equilibrium.

On the other hand, we claim that any equilibrium where the third bullet holds must be (coordinate-wise) unstable. To see this, consider such an equilibrium where $x_i = \threshold$. The only way this is a valid equilibrium state for bidder $i$ is if $\max(x_{i_1}, x_{i_2}) = \threshold$. Now, consider any arbitrarily small perturbation which decreases $x_{i_1}$ and $x_{i_2}$ while increasing $x_i$ (and perturbing other bidders in any way). By Lemma \ref{lemma:nor-gate}, it follows that $dx_i/dt > 0$ (since both $x_{i_1}$ and $x_{i_2}$ are now below $\threshold$). But we also have $x_i > \threshold$, contradicting the fact that the previous equilibrium was coordinate-wise unstable.

Finally, the number of bidders and items in this system is $O(|\mathcal{C}|)$. If we started with a general boolean network $\mathcal{C}$, the reduction of Lemma \ref{lemma:universal-nor} expanded the size of $\mathcal{C}$ to the number of NOR gates $G$ required to represent the functions $f_i$, and so the ultimate size of $\mathcal{C}$ is $O(G)$.
\end{proof}

In the case of NOR-circuits, we also get convergence to the unique fixed point of the circuit.

\begin{proof}[Proof of Theorem~\ref{thm:circuit-convergence}]
Note that since a NOR-circuit is an acyclic NOR-network, the bidder $x_i$ is completely uninfluenced by bidders $j > i$. Therefore, we can proceed inductively.

The theorem is straightforward for the base case of a single bidder (who cannot depend on any other values and therefore does not even need to converge). Now, look at the sub-dynamics $\mathcal{S'}$ of $\mathcal{S}$ corresponding to the  the circuit $\mathcal{C}'$ consisting of $X_1$ through $X_{n-1}$. By the inductive hypothesis, $\mathcal{S'}$ will converge to the unique satisfying assignment to $\mathcal{C'}$. Run $\mathcal{S'}$ long enough so that each of the bidders $x_i$ is within $\varepsilon < \min(\threshold - \low, \high - \threshold)$ of this equilibrium. Then, by the last part of Lemma \ref{lemma:nor-gate}, the full dynamics $\mathcal{S}$ (which consists of a single embedded $\NOR$ gate on top of variables in $\mathcal{S'}$, who are guaranteed to be far from $\threshold$) $x_{n}$ will also converge to either $\low$ or $\high$, as desired.
\end{proof}
}

\cameraready{}{
\subsection{Clock using a cycle of NOT gates}\label{sec:clock}
The fact that an odd cycle of NOT gates exhibits periodicity has the interesting consequence that we can use this dynamic to simulate a clock. If we consider the ROS system with $n$ variables with $n$ odd such that $X_i = \NOT(X_{i+1})$ where $X_{n+1} = X_1$ then the behavior of each of the variables resembles a clock that switches between \high and \low (Figure~\ref{fig:circuit_c}). Together with the ability to simulate any boolean circuit, this allows us to simulate any computer with finite memory inside an ROS system.

\begin{figure}[h]
\begin{center}
\includegraphics[width=\textwidth]{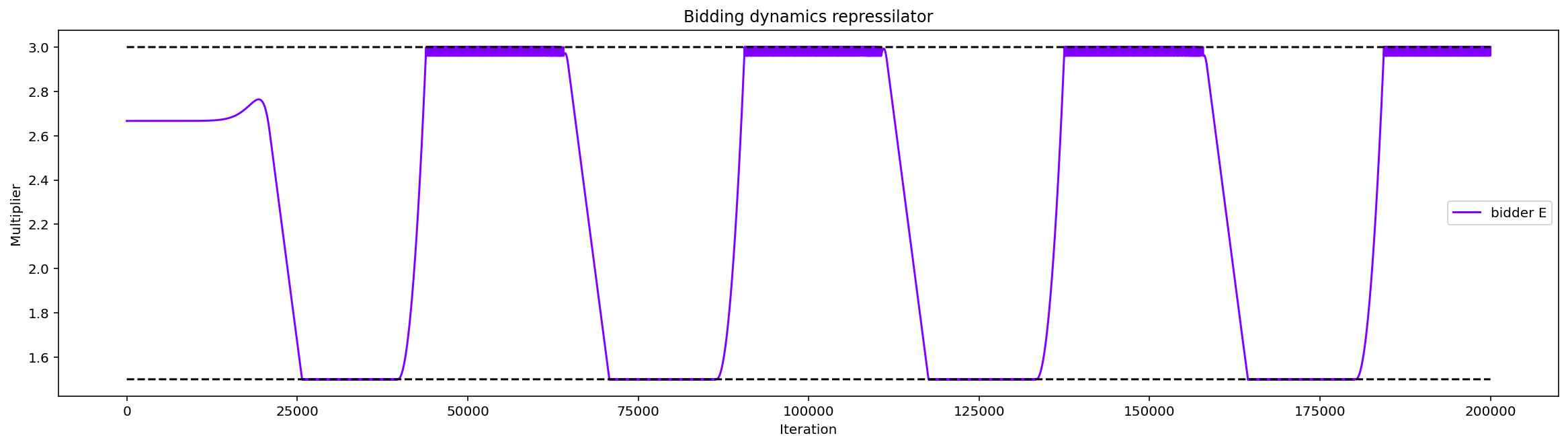}
\end{center}
\caption{Behavior of one of the bid multipliers in an ROS system encoding boolean equations $X_i = \NOT(X_{i+1})$ for $i=1\dots9$ and $X_{10} = X_1$.}
\label{fig:circuit_c}
\end{figure}
}

\end{document}